\documentclass[reqno,11pt]{amsart}
\makeatletter
\def\section{\@startsection{section}{1}%
	\z@{.7\linespacing\@plus\linespacing}{.5\linespacing}%
	{\bfseries
		\centering
}}
\def\@secnumfont{\bfseries}
\makeatother

\usepackage{amsmath}
\usepackage{amsfonts}
\usepackage{amssymb}
\usepackage{graphicx}
\usepackage{xcolor}
\usepackage{soul}
\usepackage{ulem}
\usepackage{lineno}
\newtheorem{theorem}{Theorem}[section]

\newtheorem{corollary}[theorem]{Corollary}

\newtheorem{definition}[theorem]{Definition}

\newtheorem{proposition}[theorem]{Proposition}
\newtheorem{remark}[theorem]{Remark}

\numberwithin{equation}{section}
\usepackage{lmodern}
\usepackage[colorlinks=true, allcolors=blue]{hyperref}
\colorlet{blu1}{blue!70!black}
\colorlet{blu2}{blue!50!black}
\colorlet{blu3}{blue!70!red}
\colorlet{blu4}{blue!60!green}
\colorlet{red1}{red!80}
\colorlet{red2}{red!50!black}
\colorlet{red3}{red!70!yellow}
\colorlet{red4}{red!50!yellow}
\colorlet{yel1}{yellow!50!black}
\colorlet{yel3}{yellow!20!blue}
\colorlet{gre1}{green!60!blue}
\colorlet{gre2}{green!60!black}
\colorlet{gre3}{green!40!black}
\catcode `\@=12 \textwidth16cm \textheight23cm \hoffset-1cm   \voffset-2.0cm
\begin{document}
	
\begin{center}
{\bf\Large A new interacting Fock space, the Quon algebra} 
	
{\bf\Large with operator parameter and its Wick's theorem}\vspace{1.5cm}

{\bf\large Yungang Lu}\vspace{0.6cm}
	
{Department of Mathematics, University of Bari ``Aldo Moro''}\vspace{0.2cm}
	
{Via E. Orabuna, 4, 70125, Bari, Italy}	
\end{center}\vspace{1.5cm}

\centerline{\bf\large Abstract}\bigskip\noindent
Motivated by the creation-annihilation operators in a newly defined interacting Fock space, we initiate the introduction and the study of the Quon algebra. This algebra serves as an extension of the conventional quon algebra, where the traditional constant parameter $q$ found in the $q$--commutation relation is replaced by a specific operator. Importantly, our investigation aims to establish Wick's theorem in the Quon algebra, offering valuable insights into its properties and applications.



\bigskip\bigskip

\section{Introduction}\label{qm-s1}  \bigskip

Let ${\mathcal H}$ be a (pre--)Hilbert space, and consider the quon algebra over ${\mathcal H}$ with the parameter $q\in[-1,1]$, denoted as $G_{q}(\mathcal{H})$ (see \cite{Gre91} and references therein). This algebra is defined as the $*$--algebra generated by the family $\{b_q^+(f), b_q(g): f, g\in\mathcal{H}\}$, where these $b_q^+(f)$'s and $b_q(g)$'s satisfy the standard {\it involution relation}:
\begin{equation}\label{inv-rela01}
\big(b_q(f)\big)^* = b_q^+(f) ,\quad \forall\, f\in {\mathcal H}
\end{equation}
and the {\it $q$--commutation relation}:
\begin{equation}\label{q-com-rela01}
b_q(f)b_q^+(g) - b_q^+(g)b_q(f)q = \langle f, g\rangle ,\quad \forall\, f, g\in {\mathcal H}
\end{equation}
Moreover, they can be interpreted as the creation and annihilation operators, respectively, acting on the $q$--Fock space $\Gamma_q(\mathcal{H})$ (see \cite{Bo-Kum-Spe97}, \cite{Bo-Spe91}, and references therein).

As pointed out in \cite{Gre2001}, \cite{GreHil99} and related references, the $q$--commutation relation, which is fundamental to the definition of the quon algebra, originated as a convex combination of the Canonical Commutation Relation (CCR):
\begin{equation}\label{ccr01}a(f)a^+(g)-a^+(g)a(f)=\langle f, g\rangle \,,\quad \forall\, f,g\in {\mathcal H}
\end{equation}
and the Canonical Anti--commutation Relation (CAR):
\begin{equation}\label{car01}
a(f)a^+(g)+a^+(g)a(f)=\langle f, g\rangle \,,\quad \forall\, f,g\in {\mathcal H}
\end{equation}
In fact, by combining \eqref{ccr01} with a weight of $(1+q)/2$ and \eqref{car01} with a weight of $(1-q)/2$, one finds:
\begin{equation}\label{ccr+car01}
a(f)a^+(g)-a^+(g)a(f)q=\langle f, g\rangle \,,\quad \forall\ f,g\in {\mathcal H}
\end{equation}
It is important to point out that, to the best of the author's knowledge, the first {\it systematic derivation} of the $q$--commutation relation is attributed to \cite{AcKoVo98b}.

Several variations of the $q$--commutation relation have been introduced (see \cite{AsaiYosh2020}, \cite{Blitvic2012}, \cite{BozeYosh2006}, \cite{BozeWyso2001}, \cite{Rand2019}, and references therein). In this paper, actuated by a specific example, we introduce and investigate the Quon algebra ${\mathcal Q}_{q,m}$ as a generalization of the quon algebra, in which the aforementioned convex combination appears with {\it operator coefficients}. 

Let ${\bf q}$ be a self--adjoint operator with a spectrum contained in $[-1,1]$. We get obviously that

$\bullet$ both ${1+ {\bf q}\over 2}$ and ${1- {\bf q}\over 2}$ are non--negatively defined and  bounded by the identity;

$\bullet$ ${1+ {\bf q}\over 2}+{1- {\bf q}\over 2}=$the identity. \\
This allows us to introduce the following commutation relation as a generalization of the standard $q$--commutation relation \eqref{ccr+car01}: 
\begin{equation}\label{gene-q-com-rela01}
a(f)a^+(g)-a^+(g)a(f){\bf q}=\langle f, g\rangle \,,\quad \forall\ f,g\in {\mathcal H}
\end{equation}

Due to ${\bf q}$ being an operator (not necessarily a scalar), the algebraic structure of the Quon algebra ${\mathcal Q}_{q,m}$ depends not only on the commutation relation \eqref{gene-q-com-rela01} and the standard involution relation (an analogue of \eqref{inv-rela01}, which states that ${\bf q}^*={\bf q}$ and $\big(a(f)\big)^* = a^+(f)$ for any $f\in {\mathcal H}$), but also on a new commutation relation which specifies how to {\it commune} $a(f)$ and ${\bf q}$ (or equivalently, $a^+(f)$ and ${\bf q}$). For detail, see Definition \ref{QuanAlg}.

In this paper, we begin by introducing the Quon algebra ${\mathcal Q}_{q,m}$. Then, in the main part of Section \ref{qm-s2}, we present a concrete example of the Quon algebra using a specified interacting Fock space over a given (pre--)Hilbert space. Section \ref{qm-s3} is dedicated to the establishment of Wick's theorem for the Quon algebra ${\mathcal Q}_{q,m}$, which enables us to explicitly compute the normally ordered form of any general {\it word} of ${\mathcal Q}_{q,m}$. The meaning of the term {\it word} in the context of the Quon algebra ${\mathcal Q}_{q,m}$ will be elucidated in the opening paragraph of the same section (Section \ref{qm-s3}). The main results are Theorem \ref{QuanAlg06} and Corollary \ref{QuanAlg06a9}. In the last section, we demonstrate that when ${\bf q}$ is a scalar within the interval $[-1,1]$, our Wick's theorem given in Theorem \ref{QuanAlg06} and Corollary \ref{QuanAlg06a9} is equivalent to the usual $q$--Wick's theorem, despite differences in the formulation of the $q$--weight as compared to \cite{ManSchSev2007} and related references.

\bigskip\bigskip

\section{Interacting Fock space, the $(q,m)$--Fock space and the Quon algebra}\label{qm-s2}\bigskip

First of all, we'll recall the {\it algebraic} definition of Fock space and creation operator. 

\begin{definition}\label{alg-Fock}Let $\mathcal{H}$ be a vector space. The vector space $\Gamma_0(\mathcal{H}):= \bigoplus_{n=0}\mathcal{H}^{\otimes n}$ is referred to as the {\bf algebraic Fock space} over $\mathcal{H}$, here, $\mathcal{H}^{\otimes0}:=\mathbb{C}$, $\oplus$ and $\otimes$ are interpreted as {\bf algebraic} operations. Furthermore, 

$\bullet$ the vector $\Phi:=1\oplus0\oplus0\oplus \mathcal{\ldots}$ is called the {\bf vacuum} vector of $\Gamma_0(\mathcal{H})$; 

$\bullet$ each $\mathcal{H}^{\otimes n}$ is termed the {\it $n$--particle space}.

For any $f\in \mathcal{H}$, one defines the creation operator (with the test function $f$) $A^+(f)$ as a {\bf linear} operator on the vector space $\Gamma_0(\mathcal{H})$ with the following properties:

$\bullet$ $A^+(f)\Phi:=f$;

$\bullet$ $A^+(f)(g_1\otimes\ldots\otimes g_n):=f\otimes g_1\otimes\ldots\otimes g_n$ for any $n\in\mathbb{N}^*$ and $\{g_1,\ldots, g_n\}\subset\mathcal{H}$.
\end{definition}

Secondly, let's briefly review interacting Fock space  over a given (pre--)Hilbert space $\mathcal{H}$ with scalar product $\langle \cdot,\cdot\rangle $. For any $n\in{\mathbb N}^*$, we consider a linear non--negatively defined operator $\lambda_{n}$ on $\mathcal{H}^{\otimes n}$. To simplify technical treatment, we assume additionally that $\lambda_{n}$'s are bounded, i.e., $\lambda_n \in\mathbf{B}_+\left( \mathcal{H} ^{\otimes n}\right)$ for any $n\in {\mathbb N}^*$.

We introduce the sesquilinear forms $\{\langle\cdot, \cdot \rangle_n\}_{n=1}^\infty$ as follows: 
\begin{equation}\label{sca-pro}
\langle x, y \rangle_n:= \langle x, \lambda_{n}y \rangle_{\otimes n}, \quad \forall n\in \mathbb{N}^*\text{ and }x,y\in \mathcal{H}^{\otimes n}
\end{equation}
Here and in the following discussion, $\langle\cdot, \cdot \rangle_{\otimes n}$ represents the standard $n$--fold tensor scalar product based on $\langle\cdot, \cdot \rangle$. We say that $\{\langle\cdot, \cdot \rangle_n\}_{n=1}^\infty$ is {\bf consistent} (equivalently, $\left\{\lambda_n\right\} _{n=1}^\infty$ is {\bf consistent}) if, for any $n\in\mathbb{N}^*$ and $f\in\mathcal{H}$, $\langle f\otimes x, f\otimes x \rangle_{n+1}=0$ whenever $\langle x, x \rangle_{n}=0$. Additionally, we introduce for any $n\in\mathbb{N}^*$, the following: 

$\bullet$ the relation $\sim:$ $x\sim y\ \overset {\text{def.}} {\Longleftrightarrow}\left\langle x-y,x-y\right\rangle _{n}=(x-y,\lambda_{n}\left(  x-y\right) )_ {\otimes n} =0$;
	
$\bullet$ the pre--Hilbert space $\mathcal{H}_{n}:= \left\{ \mathcal{H} ^{\otimes n}/\sim,\left\langle \cdot,\cdot\right\rangle_{n}\right\}  .$

It's worth noting that if $\{\langle\cdot, \cdot \rangle_n\}_{n=1}^\infty$ is consistent, then for any $f\in \mathcal{H}$, $A^+(f)x=0$ whenever $x=0$. Thus, $A^+(f):\mathcal{H}_n\longmapsto\mathcal{H}_{n+1}$ remains linear for any $n\in\mathbb{N}$. 

Now, based on the above definitions, for any  $f\in\mathcal{H}$, the creation operator with the test function $f\in\mathcal{H}$, i.e. $A^{+}\left( f\right)$, has the following actions:

$\bullet$ For $n=0$, it transforms $c\in\mathcal{H}_0:={\mathbb C}$ to $cf$.

$\bullet$ For $n\ge1$, it maps the equivalent class $[x]\in \mathcal{H}^{\otimes n}/\sim$ to the equivalent class $\left[ f\otimes x\right] \in\mathcal{H}^{\otimes( n+1) }/\sim$.\\ 
From this point forward, unless confusion arises, we'll denote the equivalent class $[ x] $ simply as $x$.

\begin{definition}\label{de-consistence} Let $\mathcal{H}$ be a pre--Hilbert space and $\{  \lambda_n\}_{n=1}^\infty$ be a family of operators, where $\lambda_n\in\mathbf{B}_+ \big(  \mathcal {H}^{\otimes n}\big)$ for any $n\in{\mathbb N}^*$. In the case where $\{\lambda_n\}_{n=1}^\infty$ is consistent, we refer to $\Gamma\left( \mathcal{H}, \left\{\lambda_{n} \right\}_{n}\right)$ $:= \mathcal{H}_{0}\oplus \mathcal {H}_{1}\oplus\mathcal{H}_{2}\oplus\mathcal{\ldots}$,  as the {\bf interacting Fock space} (and denoted as {\bf IFS} from now on for brevity) over $\mathcal{H}$ with the interactions $\lambda_n$'s. Hereinafter, $\mathcal{H}_{0} :=\mathbb{C}$ and the vector $\Phi:=1\oplus0\oplus0\oplus \mathcal{\ldots}$ is called the {\bf vacuum} vector of the IFS. Furthermore, for any $f\in \mathcal{H}$, if $A^+(f) \big\vert_{\mathcal{H}_n}$ is adjointable for any $n\in\mathbb{N}$, then the linear operator 
\begin{equation}\label{def_anni}
A(f):=0\oplus \bigoplus_{n=1}^\infty\big(A^+(f) \big)^*\big \vert_{\mathcal{H}_n}  
\end{equation}
is called the {\bf annihilation operator} with the test function $f$.
\end{definition}

Some well--known useful facts about IFS and creation--annihilation operators could be formulated as follows:

1) For any $f\in\mathcal{H}$ and $n\in\mathbb{N}$, $\mathcal{H}_n\overset{A^{+}( f)  } {\longmapsto} \mathcal{H}_{n+1}\overset{A(f) } {\longmapsto} \mathcal{H}_n,\ \mathcal{H}_0\left(  :=\mathbb{C} \right) \overset{A( f) } {\longmapsto}0$.

2) For any  $n\in\mathbb{N}^*$ and  $\{f_1,\ldots,f_n\}\subset\mathcal{H}$, $A^+( f_1)\ldots A^+( f_n) \Phi=f_1\otimes \ldots\otimes f_n$ (in fact, $\left[  f_1\otimes\ldots\otimes f_n\right]$, i.e.,  $\left[  A^+(f_1)  \ldots A^+(  f_n)  \Phi\right] $).
	
3) By denoting, as usual,
\begin{align}\label{QuanAlg02}
A^{\varepsilon}(f):=\begin{cases}A^{+}(f), & \text{if }\varepsilon=1\\
A(f),& \text{if }\varepsilon=-1\end{cases}\ ,\qquad \forall f\in\mathcal{H}
\end{align}
then,

$\bullet$ for any $N\in\mathbb{N}^*$ and $\varepsilon\in\{ -1, 1\}^N$, which is the set of all the functions defined on $\{  1,\ldots,N\}$ and taking values in $\{-1,1\}$, the vacuum expectation $\left\langle \Phi, A^ {\varepsilon(1)}(f_1)\ldots A^{\varepsilon(N)}(f_N) \Phi\right\rangle $ differs from $0$ only if $N=2n$ and $\varepsilon\in \{-1,1\}_+^{2n}$, where and throughout this paper,
\begin{align*}
\{-1,1\}_+^{2n}:=\Big\{\varepsilon\in\{-1,1\}^{2n}: \sum_{h=1}^{2n}\varepsilon(  h)  =0,\ \  \sum_{h=k}^{2n}\varepsilon(  h) \ge0\text{ for any } k\in\{1,\ldots,2n\} \Big\}
\end{align*}
	
$\bullet$ for any $\varepsilon\in\{ -1,1\} _+^{2n}$ and $k\in\{1,\ldots,2n\}$ such that $\sum_{h=k+1}^{2n}\varepsilon(h)=0$ (so $k$ must be even),
\begin{align*}
&\left\langle \Phi,A^{\varepsilon(1)}( f_1)\ldots A^{\varepsilon( 2n)}( f_{2n}) \Phi\right\rangle \\
=&\left\langle\Phi,A^{\varepsilon(1)}(f_1)\ldots A^{\varepsilon(k)}(f_k)\Phi\right\rangle\left\langle\Phi, A^{\varepsilon(k+1)} (f_{k+1})\ldots A^{\varepsilon(2n)}(f_{2n})  \Phi\right\rangle
\end{align*}

The above property 3) is usually called the {\bf generalized Gaussianity}.
	
Let, for any $n\in\mathbb{N}^*$, 
\[
PP(2n):=\left\{  \text{pair partitions of the set }\{  1,\ldots,2n\}\right\}
\]
\[
NCPP(2n):=\left\{  \text{non--crossing pair partitions of the set } \{1,\ldots,2n\}  \right\}
\]
Here and throughout this paper, we adopt the following conventions without loss of generality: for any $\{(l_{h},r_{h})\} _{h=1}^{n}\in PP(2n)$,
\[
l_1<l_2\,<\ldots<l_n \ \text{ and } \ \ \ l_h<r_h,\quad \forall h\in\{1,\ldots,n\}
\]
This implies that the index $l_1$ is equal to 1, and $2n$ belongs to the set $\big\{r_k:k\in\{1,\ldots,n\} \big\}$.

It is a well--established fact (see \cite{aclu-qed}, \cite{acluvo-kyoto}) that there exists a bijective relationship between the set of non--crossing pair partitions $NCPP(2n)$ and the set $\{ -1,1\}^{2n}_+$. More precisely, for any $\varepsilon\in\{-1,1\}_+^{2n} $, there exists a unique $\{(l_h,r_h)\}_{h=1}^n\in NCPP(2n)$ such that $\varepsilon(l_h) =-1$ and $\varepsilon( r_h) =1$ for all $h\in\{1,\ldots,n\}$. At times, this element in $NCPP(2n)$ is denoted as
$\{(l^\varepsilon_h,r^\varepsilon_h)\}_{h=1}^n$ to emphasize its dependence on $\varepsilon$.

\begin{remark} Let $\mathcal{H}$ be a (pre--)Hilbert space. It is worth noting that all the concrete and previously examined Fock spaces over $\mathcal{H}$ can be constructed as specific instances of an IFS with appropriately chosen interactions. Let's illustrate this with an example by defining interactions as follows:
\[
\lambda_n:=\frac{1}{n!}\sum_{\sigma\in\mathfrak{S}_{n}}P_ {\sigma},\qquad \forall n\geq2
\]
where and from this point onward, $\mathfrak{S}_n$ represents the permutation group of order $n$, and for any $\sigma\in\mathfrak{S}_n$, $P_{\sigma}$ is the {\it symmetrization operator}:
\[
P_{\sigma}\left(f_n\otimes\ldots\otimes f_1\right) :=f_{\sigma(n)}\otimes\ldots\otimes f_{\sigma(1)  },\qquad \forall\,\{f_h\}_{h=1}^n\subset\mathcal{H}
\]
Upon an easy examination, it becomes evident that:

$\bullet$ for any $n\ge2$, the space $\mathcal{H}_{n}$ is essentially equivalent to $\mathcal{H}^{\circ n}$, representing the symmetric subspace of $\mathcal{H}^{\otimes n}$;
	
$\bullet$ for any $f\in\mathcal{H}$, the operator $A^+( f)$ corresponds to the Boson creation operator with the test function $f$.
\end{remark}

When discussing a Fock space over the provided (pre--)Hilbert space $\mathcal{H}$, we use the term ``$\{ \mathcal{P}_n\}_n$--{\bf determined}'' (for general formulation, see \cite{accardi90} and \cite{accardi95}) when the following conditions hold: For any $n\in \mathbb{N}^*$ and $\{(l_h,r_h)\}_{h=1}^n\in\mathcal{P} _n$, there exists a functional $q_{\{(l_h,r_h)\}_{h=1} ^n}:\mathcal{H}^{\otimes2n} \longmapsto\mathbb{C}$ such that:
	
$\bullet$ for any $h\in\{ 1,\ldots,n\}$, $q_{\{ (l_h, r_h)\}_{h=1}^n}$ is linear with respect to the $r_h$--th variable and anti--linear with respect to the $l_h$--th variable;
	
$\bullet$ there exist $f_1,\ldots,f_{2n}\in\mathcal{H}$ such that $q_{\{(l_h,r_h)\}_{h=1}^n}(f_1,\ldots,f_{2n}) \ne 0$, in other words, $q_{\{(l_h,r_h)\}_{h=1}^n}$ is a non--zero functional;
	
$\bullet$ the following equality holds for all $\{ f_h:h\in\{1,\ldots,2n\}\} \subset\mathcal{H}$:
\begin{align}\label{QuanAlg02g}
&\langle\Phi,\big(A(f_1) +A^+(f_1)\big)\cdot\ldots\cdot \big(A(f_{2n}) +A^+(f_{2n})\big)\Phi \rangle\notag\\ =&\sum_{\{(l_h, r_h)\}_{h=1}^n \in\mathcal{P}_n} q_{\{(l_h,r_h)\}_{h=1}^n}( f_{1},\ldots, f_{2n})
\end{align}

It is important to note that all the specific Fock spaces investigated so far share the property of being $\{\mathcal{P}_{n}\}_n$--determined, with $\{\mathcal{P}_n\} _n$ falling into {\it one of the following two categories}:

$\bullet$ $\{PP(2n)\}_n$: This is the case for both Bosonic and Fermionic Fock spaces, where $ \mathcal{P}_n$ is precisely the set $PP(2n)$ for all $n$. In general, the $q$--Fock space (as discussed in \cite{Bo-Spe96}, \cite{Bo-Kum-Spe97}, \cite{Fris-Bou70}, \cite{Gre91}, and related references) with $q\ne0$ also falls into this category.

$\bullet$ A subset of $NCPP(2n)$: For various ``free-type'' IFS introduced in \cite{lu95}, which includes examples like the free Fock space, the Boolean Fock space, the monotone Fock space, the chromatic Fock space, and 1--mode type IFS, they are $\{ \mathcal{P}_n\}_n$--determined, where $\mathcal{P}_n$ is a subset of the set $NCPP(2n)$ for any $n$. More concretely:

---for the free (as well as the monotone, the chromatic and 1--mode type) Fock space, $\mathcal{P}_n=NCPP(2n)$;

---for the Boolean Fock space, $\mathcal{P}_{n}$ consists solely of  $\{(2h-1,2h)\} _{h=1}^n$.

It would be valuable to illustrate a specific example of a Fock space that is $\{\mathcal{P}_n\}_n$--determined, while ensuring that the set $\mathcal{P}_n$ is distinct from both $PP(2n)$ and a subset of $NCPP(2n)$ for any $n\ge3$. In other words, we aim to construct a specific $\{\mathcal{P}_n\}_n$--determined Fock space that satisfies the following inequalities for any $n\ge3$,
\begin{equation}\label{QuanAlg02h}
NCPP(2n)\subsetneqq\mathcal{P}_n\subsetneqq PP(2n) 
\end{equation}
It's essential to emphasize that this condition applies only for $n\ge3$, as the following observations confirm the non--existence of such $\mathcal{P}_n$ for $n\in\{1,2\}$:

$\bullet$ for $n=1$, the fact $NCPP(2)=PP(2)$ indicates that there is no intermediate $\mathcal{P}_1$ that satisfies \eqref{QuanAlg02h};

$\bullet$ for $n=2$, one observes that $\left\vert NCPP(4) \right \vert =2$ and $\left\vert PP(4)\right\vert =3$, which further underscores the absence of $\mathcal{P}_2$ such that \eqref{QuanAlg02h} holds.

Notice that if our Fock space is Gaussian type (see \cite{lu-gifs}) with the weight function $\omega: \cup_{n}PP(2n)\longmapsto\mathbb{C}$, i.e., for all $n\in{\mathbb N}^*$ and $\{f_1,\ldots, f_n\}\subset{\mathcal H}$,
\[
q_{\{(l_h,r_h)\}_{h=1}^n}(f_1,\ldots,f_{2n})  :=\omega\left(\{ (l_h,r_h)\}_{h=1}^n\right) \prod_{h=1}^n\big\langle f_{l_h},f_{r_h}\big\rangle
\]
then it must be
$\{\mathcal{P}_n\}_n$--determined by taking
\[
\mathcal{P}_n:=\big\{ \{ (l_h,r_h)\}  _{h=1}^n\in PP(2n):\omega( \{ (l_h,r_h) \} _{h=1}^n) \ne0\big\}
\]

In general, a $\{ \mathcal{P}_n\} _n$--determined Fock space is not necessarily Gaussian in nature. The monotone Fock space supplies such an example: it is $\{NCPP(2n)\} _n$-- determined but not of Gaussian type, as shown in Section $\mathcal{x}4$ of \cite{lu-gifs}.

Now, let's consider a particular IFS. Suppose $\mathcal{H}$ is a vector space, and let $m\in\mathbb{N}^*\cup {+\infty}$ and $q\in[-1,1]$. We introduce a sequence of linear operators $\{\lambda_n\}_{n=1}^\infty$, which will be referred to as {\bf $(q,m)$--interactions}, such that:

$\bullet$ $\lambda_1={\bf 1}:=$the identity operator on $\mathcal{H}$ and for any $n> m$, $\lambda_n:={\bf 1}^{\otimes (n-m)} \otimes \lambda_m$;

$\bullet$ for any $1\le n\le m$, $\lambda_n$ is the $q$--{\it symmetrization} operator on $\mathcal{H} ^{\otimes n}$, which is defined as follows:
\begin{equation}\label{symm-op}
\lambda_n(g_1\otimes \ldots\otimes g_n):= \sum_{\sigma\in \mathfrak{S}n}q^{\rho(\sigma)} g_{\sigma(1)}\otimes \ldots\otimes g_{\sigma(n)},\quad \forall\{ g_1,\ldots,g_n\} \subset\mathcal{H}
\end{equation}
hereinafter, for any $\sigma\in\mathfrak{S}_n$, $\rho(\sigma)$ is defined as the {\it number of inversions} of $\sigma$:
\begin{equation}\label{inver-nu}
\rho(\sigma):=\big\vert\{(i,j):1\le i<j\le n,\,\sigma(i)>\sigma(j)\} \big\vert
\end{equation}
As shown by various authors (e.g., \cite{Bo-Spe91}, \cite{Fiv92}, \cite{Zag92}) that $\lambda_n$ is non--negatively defined for any $n\in{\mathbb N}^*$. 

Now, we introduce the {\bf $(q,m)$--annihilation operator}. To define this operator, we require a scalar product $\langle\cdot,\cdot\rangle:\mathcal{H} \times\mathcal{H} \longrightarrow \mathbb{C}$, which endows $\mathcal{H}$ with a pre--Hilbert space structure. 

\begin{definition}\label{alg-anni}
Let $\mathcal{H}$ be a (pre--)Hilbert space equipped with the scalar product $\langle\cdot,\cdot\rangle$, and let $q\in[-1,1]$. For any $m\in\mathbb{N}^*\cup\{+\infty\}$ and $f\in\mathcal {H}$, the $(q,m)$--annihilation operator $a_{q,m}(f): \Gamma_0(\mathcal{H}) \longrightarrow\Gamma_0(\mathcal{H})$ is defined based on {\bf linearity} and further characterized as follows:
	
$\bullet$ $a_{q,m}(f)\Phi:=0$;
	
$\bullet$ for any $n\in\mathbb{N}^*$ and $\{g_1,\ldots, g_n\}\subset\mathcal{H}$,
\begin{align}\label{alg-crea01}
&a_{q,m}(f)(g_1\otimes\ldots\otimes g_n)\notag\\
:=&\begin{cases}
\sum_{k=1}^nq^{k-1}\,\langle f, g_k\rangle g_1\otimes\ldots \otimes g_{k-1}\otimes g_{k+1}\otimes\ldots\otimes g_n,&\text{ if }n< m+1;\\
\langle f, g_1\rangle g_2\otimes\ldots\otimes g_n,&\text{ if }n\ge m+1  \end{cases}
\end{align}
\end{definition}

\begin{remark} Clearly, the operator $a_{q,m}(f)$ defined above corresponds to  
	
$\bullet$ the standard $q$--annihilation if $m=+\infty$;

$\bullet$ the free annihilation operator if either $q=0$ or $m=1$.\\
Therefore, genuinely interesting case arises when $1< m<+\infty$ and $q\ne0$.
\end{remark}

In the following, the IFS $\Gamma\left( \mathcal{H},\{ \lambda_n\}_n\right)$ will be referred to as the {\bf $(q,m)$--Fock space} if $\{ \lambda_n\}_n$ is the $(q,m)$--interactions; the operators $A^+(f)$ and $a_{q,m}(f)$ introduced in Definition \ref{alg-Fock} and Definition \ref{alg-anni}, are termed the $(q,m)$--creation and $(q,m)$--annihilation operators, respectively, with the test function $f\in\mathcal{H}$.

\begin{proposition}\label{DCT05} For any $m\in\mathbb{N}^*$ and $q\in[-1,1]$, for any (pre--)Hilbert space $\mathcal{H}$, the following statements are true for all $f$ belonging to $\mathcal{H}$ when considering the $(q,m)$--Fock space $\Gamma\left( \mathcal{H},\{\lambda_n\}_n\right)$: 
	
1)  $(q,m)$--creation operator $A^+(f)$ is bounded and its norm is given by:
\begin{align}\label{DCT05g0}
\Vert A^+(f)\Vert =\Vert f\Vert\cdot\begin{cases} \sqrt{1+q+\ldots+q^{m-1}},&\text{ if }q\in[0,1];\\ 1,&\text{ if }q\in[-1,0)  \end{cases}
\end{align}
In particular, its restriction to any $n$--particles space $\mathcal{H}_n$ is bounded.

2) $A(f):=\big(A^+( f)\big)^*$ is well defined and bounded.
	
3) The $(q,m)$--annihilation operator $a_{q,m}(f)$ is equal to $A(f)$.
\end{proposition}
	
\begin{proof} The assertion 2) is a trivial consequence of the assertion 1).

By representing the $(q,m)$--Fock space $\Gamma\big(  \mathcal{H},\{ \lambda_n\}_n\big)  $ as the direct sum
\[
\big( \mathbb{C}\oplus\mathcal{H}\oplus\ldots\oplus \mathcal{H}_m\big)\oplus\bigoplus_{n=m+1}^{\infty} \mathcal{H}_n
\]
then, $A^{+}(f)$ can be expressed as the direct sum of two operators, $b_{q,m}^+(f)$ and $b^+(f)$, where $b_{q,m}^+(f)$ and $b^+(f)$ represent the restrictions of the creation operator with the test function $f$ to:
$$\bigoplus_{n=0}^m \big(n-\text{particles space of the $q$--Fock space}\big)$$
and
$$\bigoplus_{n=m+1}^\infty\big(n-\text{particles space of the free Fock space}\big)$$
respectively. In this way, $A^{+}(f)$ is decomposed into these two components. Furthermore, it is well known (see, e.g. \cite{Bo-Kum-Spe97}, \cite{Bo-Spe91}, \cite{Bo-Spe96}, \cite{JiKim2006}) that both $b_{q,m}^+(f)$ and $b^{+}(f)$ are bounded: for any $f\in{\mathcal H}$,
\begin{align}\label{DCT05g0x}
\Vert b^+(f)\Vert=&\Vert f\Vert \notag\\ 
\Vert b^+_{q,m}(f)\Vert =&\Vert f\Vert\cdot \begin{cases} \sqrt{1+q+\ldots+q^{m-1}},&\text{ if }q\in[0,1];\\ 1,&\text{ if }q\in[-1,0)  \end{cases}
\end{align}
and consequently, 
\[\Vert A^+(f)\Vert=\Vert b^+_{q,m}(f)\oplus b^+(f)\Vert=max\{\Vert b^+_{q,m}(f)\Vert,\ \Vert b^+(f)\Vert\}\]
So, \eqref{DCT05g0} is obtained by combining this formula with \eqref{DCT05g0x}.

Now, let's proceed to prove assertion 3), which involves demonstrating that 
\begin{equation}\label{le2-1a0}
\big\langle \Phi,a_{q,m}( f)g\big\rangle=\langle f,g\rangle, \quad\forall f,g\in\mathcal{H}
\end{equation}
and for any $n\ge2$ , $G\in \mathcal{H}^{\otimes (n-1)}$ and $\{f,f_1,\ldots,f_n\}\subset\mathcal{H}$,
\begin{equation}\label{le2-1a1}
\big\langle G,a_{q,m}( f)(f_1\otimes\ldots\otimes f_n)\big\rangle=\big\langle f\otimes G, \lambda_n (f_1\otimes\ldots\otimes f_n)\big\rangle_{\otimes n}
\end{equation}
	
The definition of $a_{q,m}(f)$, i.e. \eqref{alg-crea01}, gives evidently \eqref{le2-1a0}.

For any $n\in\{2,\ldots ,m\}$, we have
\begin{align}\label{symm-op01}
&\big\langle f\otimes G,\lambda_n(f_1\otimes\ldots \otimes f_n)\big\rangle_{\otimes n}\overset{ \eqref{symm-op}}=\sum_{\sigma \in \mathfrak{S}_n} q^{\rho(\sigma)} \big\langle f\otimes G, f_{\sigma(1)} \otimes \ldots\otimes f_{\sigma(n)} \big\rangle_{\otimes n}\notag\\
=&\sum_{k=1}^n\sum_{\sigma\in \mathfrak{S}_{n,k}}
q^{\rho(\sigma)}\langle f, f_{k}\rangle \big\langle G, f_{\sigma(2)}\otimes \ldots\otimes f_{\sigma(n)}\big\rangle_{\otimes (n-1)}  
\end{align}
hereinafter, 
\[\mathfrak{S}_{n,k}:=\big\{ \sigma\in \mathfrak{S}_n:\ \sigma(1)=k\big\},\qquad \text{for any $n\in\mathbb{N}^*$ and $k\in\{1,\ldots,n\}$} \]

For any $\sigma\in\mathfrak{S}_{n,k}$, we define the function $\sigma_k$ as follows:
\[\sigma_k(j):=
\begin{cases} \sigma(j+1),&\text{ if }\sigma(j+1)<k\\
\sigma(j+1)-1,&\text{ if }\sigma(j+1)>k
\end{cases}
\]
Furthermore, for any $k\in\{1,\ldots,n\}$ and $j\in\{1,\ldots,n-1\}$, for any $\{f_1,\ldots, f_n\} \subset \mathcal{H}$, we introduce the vectors $g_j$'s as follows:
\[g_j:= \begin{cases} f_j,&\text{ if }j<k\\
f_{j+1},&\text{ if }k\le j\le n-1    \end{cases}
\]
It is easy to see that for any $n\in\mathbb{N}^*$ and $k\in\{1,\ldots,n\}$, hold the following properties:

$\bullet$ $(\sigma(2),\ldots,\sigma(n))= (\sigma_k(1),\ldots, \sigma_k(n-1))$;
	
$\bullet$ $\sigma_k\in \mathfrak{S}_{n-1}$, and the map $\sigma\longmapsto \sigma_k$ is a bijection from $\mathfrak{S}_{n,k}$ to $\mathfrak{S}_{n-1}$;
	
$\bullet$ $\rho (\sigma)=\rho (\sigma_k)+k-1$;
	
$\bullet$ $(g_1,\ldots,g_{n-1})=(f_1,\ldots,f_{k-1}, f_{k+1},\ldots,f_n)$ and $f_{\sigma(2)}\otimes \ldots\otimes f_{\sigma(n)}=g_{\sigma_k(1)}\otimes \ldots\otimes g_{\sigma_k(n-1)}$.\\
With these notations and properties in place,  \eqref{symm-op01} can be continued as follows by using the induction argument:
\begin{align}\label{symm-op02}
&\big\langle f\otimes G, \lambda_n (f_1\otimes \ldots\otimes f_n)\big\rangle_{\otimes n}\notag\\
=&\sum_{k=1}^nq^{k-1}\left\langle f, f_{k}\right \rangle\sum_{\tau\in\mathfrak{S}_{n-1}}q^{\rho(\tau)} \left\langle G, g_{\tau(1)}\otimes \ldots\otimes g_{\tau(n-1)}\right\rangle_{\otimes (n-1)} \notag\\
=&\sum_{k=1}^nq^{k-1}\langle f,f_k\rangle\big\langle G,\lambda_{n-1}\big(g_1\otimes \ldots\otimes g_{n-1}\big)\big\rangle_{\otimes (n-1)} \notag\\
=&\sum_{k=1}^nq^{k-1}\langle f,f_k\rangle\big\langle G,g_1\otimes\ldots\otimes g_{n-1}\big\rangle\notag\\
=&\sum_{k=1}^nq^{k-1}\langle f,f_k\rangle\big\langle G, f_1\otimes \ldots\otimes f_{k-1}\otimes f_{k+1}\otimes \ldots\otimes f_n\big\rangle\notag\\
\overset{\eqref{alg-crea01}}{=}& \big\langle G,a_{q,m}(f)(f_1\otimes\ldots\otimes f_n)\big\rangle 
\end{align}

Finally, let's consider the case $n>m$. In this scenario, instead of \eqref{symm-op01}, we have
\begin{align}\label{symm-op03}
&\big\langle f\otimes G,\lambda_n (f_1\otimes \ldots \otimes f_n)\big\rangle_{\otimes n}=\big\langle f\otimes G, f_1\otimes\lambda_{n-1} (f_2\otimes \ldots\otimes f_n)\big\rangle_{\otimes n}\notag\\
=&\langle f,f_1\rangle \big\langle G, \lambda_{n-1} (f_2\otimes\ldots\otimes f_n)\big\rangle_{\otimes (n-1)} =\langle f,f_1\rangle \big\langle G, f_2\otimes\ldots\otimes f_n\big\rangle
\end{align}
This expression is, in accordance with Definition \ref{alg-anni}, the same as the scalar product $\big\langle G,a_{q,m}( f)$ $(f_1\otimes\ldots \otimes f_n)\big\rangle $. \end{proof}

\begin{proposition}\label{(q,m)-comm01} Let $\mathcal{H}$ be a (pre--)Hilbert space, let $m\in\mathbb{N} ^*$ and $q\in[  -1,1]$. By introducing \begin{equation}\label{(q,m)-comm01c}
P_{[0,k)}:=\begin{cases}0,&\text{ if }k=0\\
\text{the projector onto }\bigoplus_{j=0} ^{k-1}\mathcal{H}_j, &\text{ if }k\ge 1	
\end{cases},\qquad \forall k\in \mathbb{N}^*
\end{equation} 
we have the following equalities:
\begin{equation}\label{(q,m)-comm01a}
A(f)A^+(g)-	qA^+(g)A(f)\, P_{[0,m)}=\langle f, g\rangle \,,\quad \forall\, f,g\in\mathcal{H}
\end{equation}
and 
\begin{equation}\label{(q,m)-comm01b}
A(f)P_{[0,k+1)}=P_{[0,k)}A(f)\,,\quad \forall\, k\in \mathbb{N} \text{ and } f\in\mathcal{H}
\end{equation}
\end{proposition}

\begin{remark} 1) The formula \eqref{(q,m)-comm01a} is evidently a generalization of the standard $q$--commutation relation. \\
2) As $A(f)P_{[0,m)}\Phi=A(f)\Phi=0$, one can substitute $P_{[0,m)}$ in \eqref{(q,m)-comm01a} with $P_{(0,m)}$ which is defined as the projector onto  $\bigoplus_{n=1}^m\mathcal{H}_n$. 
\end{remark}

\begin{proof} {\bf (of Proposition \ref{(q,m)-comm01})} 
The equality \eqref{(q,m)-comm01a} is surely equivalent to  asserting
\begin{equation}\label{(q,m)-comm01d}
\big(A(f)A^+(g)-qA^+(g)A(f)\, P_{[0,m)}\big)\big(f_1\otimes \ldots\otimes f_n\big) =\langle f, g\rangle\ f_1\otimes \ldots\otimes f_n
\end{equation}
for any $n\in\mathbb{N}$ and $\{f_1,\ldots,f_n\}\subset \mathcal{H}$. Now we proceed to prove this new formulation.

In the case of $n< m$, it holds that $P_{[0,m)}G=G$ for any $G\in \mathcal{H}_n$, and consequently, $A^+(g)G\in\bigoplus _{k=0}^m\mathcal{H}_k$. So
\begin{align}
&A(f)A^+(g)\big(f_1\otimes \ldots\otimes f_n\big) =A(f)\big(g\otimes f_1\otimes \ldots\otimes f_n\big)\notag\\
=&\langle f, g\rangle f_1\otimes \ldots\otimes f_n+
\sum_{k=1}^nq^{k}\,\langle f, f_k\rangle g\otimes f_1\otimes\ldots \otimes f_{k-1}\otimes f_{k+1}\otimes\ldots\otimes f_n\label{(q,m)-comm01e}
\end{align}
and
\begin{align}
&qA^+(g)A(f)\,P_{[0,m)}\big(f_1\otimes\ldots\otimes f_n\big)
=qA^+(g)A(f)\big(f_1\otimes \ldots\otimes f_n\big)\notag\\
=&q\sum_{k=1}^nq^{k-1}\langle f,f_k\rangle g\otimes f_1\otimes\ldots \otimes f_{k-1}\otimes f_{k+1}\otimes\ldots\otimes f_n
\label{(q,m)-comm01f}
\end{align}
Equality \eqref{(q,m)-comm01d} is obtained by combining \eqref{(q,m)-comm01e} and \eqref{(q,m)-comm01f}.

In the case of $n\ge m$ (i.e., $n+1> m$), it holds that $P_{[0,m)}G=0$ for any $G\in \mathcal{H}_n$, and consequently, $A^+(g)G\in\bigoplus _{k=m+1}^\infty\mathcal{H}_k $. Therefore \eqref{(q,m)-comm01d} is derived as follows:
\begin{align}
&\big(A(f)A^+(g)-qA^+(g)A(f)\, P_{[0,m)}\big)\big(f_1\otimes \ldots\otimes f_n\big) =A(f)A^+(g)\big(f_1\otimes \ldots\otimes f_n\big)\notag\\
=&A(f)g\otimes f_1\otimes \ldots\otimes f_n\overset{n+1>m}{=}\langle f, g\rangle\ f_1\otimes \ldots\otimes f_n\label{(q,m)-comm01g}
\end{align}

Now let's proceed to prove \eqref{(q,m)-comm01b}. In fact, we'll show that for any $n\in \mathbb{N}$ and $G_n\in \mathcal{H}_n$
\begin{equation}\label{(q,m)-comm01h}
A(f)P_{[0,k+1)}G_n=P_{[0,k)}A(f)G_n\,,\quad \forall\, k\in \mathbb{N}^*\text{ and }f\in{\mathcal H}
\end{equation}

If $n\ge k+1$ (equivalently, $n-1\ge k$), both sides of the equality in \eqref{(q,m)-comm01h} are zero because

$\bullet$ $P_{[0,k+1)}G_n=0$;

$\bullet$ $A(f)G_n\in\mathcal{H}_{n-1}$ and so $P_{[0,k)}A(f)G_n=0$.

If $n\le k$ (equivalently, $n< k+1$), both sides of the equality in \eqref{(q,m)-comm01h} are $A(f)G_n$ because

$\bullet$ $P_{[0,k+1)}G_n=G_n$ and so $A(f)P_{[0,k+1)}G_n=A(f)G_n$;

$\bullet$ $A(f)G_n\in\mathcal{H}_{n-1}$ and so  $P_{[0,k)}A(f)G_n=A(f)G_n$.\\
Summing up, the proof is completed. \end{proof}

Proposition \ref{(q,m)-comm01} suggests that we should generalize the standard $q$--commutation relation and the quon algebra.

\begin{definition}\label{QuanAlg}
Let $q\in [-1,1]$ and $m\in{\mathbb N}^*\cup \{+\infty\}$, and let ${\mathcal H}$ be a (pre--)Hilbert space. We refer to the *-algebra generated by $\{A(f),A^+(g),{\bf p}_k:\, f,g\in {\mathcal H},\,k\in {\mathbb Z}\}$ as the {\bf Quon algebra} with the parameter $(q,m)$ (commonly denoted as ${\mathcal Q}_{q,m}$), here, $A(f)$'s, $A^+(g)$'s and ${\bf p}_k$'s satisfy the following properties:
\begin{align}\label{QuanAlg01a}
{\bf p}_k={\bf p}^*_k,\qquad {\bf p}_k{\bf p}_h={\bf p}_h{\bf p}_k={\bf p}_k,\qquad \forall k\le h
\end{align}
\begin{align}\label{QuanAlg01b}
A^+(f)=\big(A(f)\big)^* \text{ and } A(f){\bf p}_{k+1}={\bf p}_kA(f)\,,\qquad \forall\, k\in \mathbb{Z} \text{ and } f\in\mathcal{H}
\end{align}
and 
\begin{align}\label{QuanAlg01c}
A(f)A^+(g)-qA^+(g)A(f){\bf p}_{m}=\langle f,g\rangle\,,\qquad \forall\,  f,g\in\mathcal{H}
\end{align}
\end{definition}\medskip

In this paper, we'll refer to the following:

$\bullet$ the equality in \eqref{QuanAlg01c} as the {\bf first $(q,m)$--commutation relation};

$\bullet$ the second equality in \eqref{QuanAlg01b} as the {\bf second $(q,m)$--commutation relation};

$\bullet$ the second equality in \eqref{QuanAlg01a} as the {\bf third $(q,m)$--commutation relation};

$\bullet$ $A(f)$ (respectively, $A^+(f)$) as the {\bf annihilation} (respectively, {\bf creation}) operator with the test function $f\in{\mathcal H}$. 

\begin{remark} 1) The first $(q,m)$--commutation relation \eqref{QuanAlg01c} requires that ${\mathcal Q}_{q,m}$ is unital. Additionally, we'll adopt the conventions ${\bf a}^0:=1$, $\prod_{j\in\emptyset}{\bf a}_j=1$ and $\sum_{j\in\emptyset}{\bf a}_j=0$.

2) Due to \eqref{QuanAlg01a} and \eqref{QuanAlg01b}, we obtain
\begin{align}\label{QuanAlg01d}
&{\bf p}_{k+1}A^+(f)=A^+(f){\bf p}_k\,,\qquad\forall \, k\in\mathbb{Z}\text{ and }f\in\mathcal{H}\notag\\
&A^+(g)A(f){\bf p}_{k+1}={\bf p}_{k+1}A(f)A^+(g), \qquad\forall \, k\in\mathbb{Z}\text{ and }f,g\in\mathcal{H}\notag\\
&A(f)A^+(g){\bf p}_{k}={\bf p}_{k}A(f)A^+(g), \qquad \forall \, k\in\mathbb{Z}\text{ and }f,g\in\mathcal{H}
\end{align}
In particular, the expression on the left--hand side of the equality in \eqref{QuanAlg01c} is equal to $A(f)A^+(g)-q{\bf p}_{m}A^+(g)A(f)$. Furthermore, by applying the second $(q,m)$--commutation relation (i.e., the second equality in \eqref{QuanAlg01b}) and its conjugate form (i.e. the first equality in \eqref{QuanAlg01d}), an induction's argument ensures that for any $n\in{\mathbb N}^*$, $\{f_1,\ldots,f_n\}\subset{\mathcal H}$ and $\varepsilon\in \{-1,1\}^n$,
\begin{align}\label{QuanAlg01f}
{\bf p}_kA^{\varepsilon(1)}(f_1)\ldots A^{\varepsilon(n)} (f_n)= A^{\varepsilon(1)}(f_1) \ldots A^{\varepsilon(n)}(f_n) {\bf p}_{k-\varepsilon(1)-\ldots- \varepsilon(n)}
\end{align}   \end{remark}

As demonstrated in Proposition \ref{(q,m)-comm01},
it becomes evident that $G_{q,m}(\mathcal{H})$ serves as a concrete example of the Quon algebra with the parameter $(q,m)$. In this algebra,

$\bullet$ we introduce $P_{[0,k)}$ following the prescription in \eqref{(q,m)-comm01c} for any $k\in {\mathbb N}$ and set $P_{[0,k)}:=0$ for any $k<0$;

$\bullet$ the symbols $A^+(f)$ and $A(f)$ respectively  represent the $(q,m)$--creation and $(q,m)$--annihilation operators associated with the test function $f\in\mathcal{H}$.

At the end of this section, we provide two remarks. One of them pertains to the construction of the set ${\mathcal P}_n$, while the other concerns the {\it vacuum distribution} of the {\it field operator} $A(f)+A^+(f)$ with $0\ne f\in{\mathcal H}$. 

To understand the construction of the set ${\mathcal P}_n$, we introduce the set ${\mathcal P}_n(\varepsilon)$, defined as follows, for any $n\ge 3$ and $\varepsilon\in\{-1,1\}_+^{2n}$:
\begin{align}\label{QuanAlg08}
\Big\{\{(l_h,r_h)\}_{h=1}^n\in {\mathcal P}_n:\,\varepsilon^{-1}(\{-1\}) =\big\{l_h:h\in\{1,\ldots,n\}\big\}\Big\}
\end{align} 
It is evident that ${\mathcal P}_n$ is the union of all ${\mathcal P}_n(\varepsilon)$ for $\varepsilon \in\{-1,1\}_+ ^{2n}$, and further,
${\mathcal P}_n(\varepsilon)\cap {\mathcal P}_n(\varepsilon') =\emptyset$ whenever $\varepsilon\ne \varepsilon'$. 

One of main result of \cite{lu2023c} pertains to the construction of ${\mathcal P}_n(\varepsilon)$. To state this result, we recall that (see \cite{aclu-qed}), for $\varepsilon\in\{-1,1\}_+^{2n}$, there exists a unique element in $NCPP(2n)$, denoted as $\{(l^\varepsilon_h,r^\varepsilon_h)\}_{h=1}^n$, such that $\varepsilon^{-1}(\{-1\})=\big\{l_h:h\in\{1,\ldots ,n\}\big\}$. Let's introduce, for any $\theta:=\{(l_h,r_h)\}_{h=1}^n\in NCPP(2n)$, and for any $k\in\{1,\ldots,n\}$, the concept of {\it depth} of the pair $(l_k,r_k)$ in $\theta$ as follows:
\[d(l_k,r_k):=\big\vert\{h:l_h<l_k<r_k<r_h\}\big\vert\]
The construction of ${\mathcal P}_n(\varepsilon)$ is strongly related to the depth of the pairs within the non--crossing pair partition $\{(l^\varepsilon_h, r^\varepsilon_h) \}_{h=1}^n$. In fact, it has been proved in \cite{lu2023c} that any element of ${\mathcal P}_n(\varepsilon)$ consists of:

$\bullet$ all pairs $(l^\varepsilon_h,r^\varepsilon_h)$ with the depth greater than or equal to 2;

$\bullet$ a generic pair partition of the set $ \big\{l^\varepsilon_h,r^\varepsilon_h: d(l^\varepsilon_h,r^\varepsilon_h)<m\big\}$.

Finally, for any $0\ne f\in{\mathcal H}$, in \cite{lu2023d}, we have calculated $\mu_q$, which is the vacuum distribution of the field operator $A(f)+A^+(f)$ when our Fock space is $(q,2)$--Fock space. This distribution strongly depends on the value of $q$:

$\bullet$ in the case of $q=-1$, $\mu_q$ is the two point distribution over $\{ -\Vert f\Vert ^2,\Vert f\Vert^2\}  $ with the equi--probability $\frac{1}{2}$;

$\bullet$ in the case of $q=0$, $\mu_q$ is the semi--circle distribution on the interval $(  -2\Vert f\Vert ,2\Vert f\Vert )$;

$\bullet$ for any $q$ in the interval $\left( -1,\frac{1}{2}\right]$, $\mu_q$ is absolutely continuous, and its probability density function is determined;

$\bullet$ for any $q$ in the interval $\left(  \frac{1}{2},1\right]$,  $\mu_q$ is {\bf not} absolutely continuous; in fact it takes the form of $\nu_q+a\big(\delta_b+ \delta_{-b}\big)$, where,

\ \ ---$\nu_q$ is absolutely continuous and its density 
function (although not a probability density function) is computed;

\ \ ---$a$ and $b$ are positive constant that are determined;

\ \ ---$\delta_x$ represents the Dirac measure concentrated at $x\in{\mathbb R}$.

\bigskip

\section{Normally ordered form of a general word of Quon algebra}\label{qm-s3}

The primary objective of this section is to establish a crucial structural theorem--Wick's theorem within the Quon algebra framework. In essence, our goal is to ascertain the normally ordered representation of any generic {\bf word} in the Quon algebra. Throughout this paper, we define a {\it word in the Quon algebra} as a product comprising a finite number of creation--annihilation operators, as well as ${\bf p}_k$'s, specifically, such elements taking the form 
\begin{align}\label{QuanAlg02a}
{\bf p}^{\epsilon(1)}_{k_1}A^{\varepsilon(1)}(f_1) {\bf p}^{\epsilon(2)}_{k_2}A^{\varepsilon(2)}(f_2) \ldots {\bf p}^{\epsilon(n)}_{k_n}A^{\varepsilon(n)} (f_n){\bf p}^{\epsilon(n+1)}_{k_{n+1}}
\end{align}
Here, 

$\bullet$ $n\in{\mathbb N}^*$, $\{f_1,\ldots, f_n\}\subset {\mathcal H}$, $\{k_1,\ldots, k_n\}\subset {\mathbb N}$,  $\varepsilon\in\{-1,1\}^n$ and $\epsilon\in\{0,1\}^n$;

$\bullet$ for any $k\in {\mathbb N}$ such that ${\bf p}_k\ne0$, we define ${\bf p}_k^\epsilon$ as follows: ${\bf p}_k^\epsilon:=\begin{cases}
{\bf p}_k, &\text{ if }\epsilon=1\\ 1, &\text{ if }\epsilon=0\end{cases}$.

Let $X$ be a set with a cardinality of at least 2, and let $\{x_1,\ldots,x_n\}\subset X$ with a cardinality of $n\ge2$, i.e., $x_k$'s are pairwise distinct. We define the following set: 
\begin{align}\label{QuanAlg02x}
X^\infty&:=\{\text{maps from }{\mathbb Z}\text { to } X\}\notag\\
X^n&:=\{\text{maps from }\{1,\ldots,n\}\text { to } X\},\quad \forall n\in{\mathbb N}^*
\end{align}
Additionally, for any given $x_0\in X$ and $n\in{\mathbb N}^*$, we can embed $X^n$ into $X^\infty$ by identifying $\rho\in X^n$ with the following $\rho'\in X^\infty$:
\begin{align}\label{QuanAlg02y}
\rho'(k):=\begin{cases}\rho(k),&\text{ if }k\in\{1,\ldots,n\}\\ x_0,&\text{ otherwise }\end{cases}
\end{align}
Moreover, throughout this paper, we'll adopt the following notations for any pair of integers $n$ and $N$ such that $n\leq N$:
\begin{align}\label{QuanAlg02z}
&[n,N]:=\{k:\,n\le k\le N\};\quad (n,N):=\{k:\,n< k< N\}\notag\\          &(n,N]:=\{k:\,n< k\le N\};\quad [n,N):=\{k:\,n\le k< N\} 
\end{align}

The normally ordered form of a general word in the Quon algebra typically involves arranging the operators in a specific order so that creation operators come before annihilation operators. More precisely, in the algenera ${\mathcal Q}_{q,m}$, a word is considered {\bf normally ordered} if it has the form
\begin{align}\label{QuanAlg02b}
{\bf p}^{\epsilon}_{k}\cdot(\text{a product of finite many creators})\cdot (\text{a product of finite many annihilators})
\end{align}
 
Using \eqref{QuanAlg01f}, we can rewrite the word \eqref{QuanAlg02a} as follows:
\begin{align}\label{QuanAlg02c}
&{\bf p}^{\epsilon(1)}_{k_1} {\bf p}^{\epsilon(2)} _{k_2+ \varepsilon(1)}{\bf p}^{\epsilon(3)}_{k_3+ \varepsilon(1)+ \varepsilon(2)}\ldots{\bf p} ^{\epsilon(n+1)}_{k_{n+1} +  \varepsilon(1)+\ldots+ \varepsilon(n)}\notag\\
&\cdot A^{\varepsilon(1)}(f_1) A^{\varepsilon (2)}(f_2) \ldots A^{\varepsilon(n)} (f_n) 
\end{align}
Where, in accordance with \eqref{QuanAlg01a},
the word \eqref{QuanAlg02c} equals

$\bullet$ in the case of $\epsilon(j)=0$ for all $j\in[1,n+1]$, to
\begin{align}\label{QuanAlg02d}
A^{\varepsilon(1)}(f_1) A^{\varepsilon(2)}(f_2) \ldots A^{\varepsilon(n)} (f_n)
\end{align}

$\bullet$ in alternative case (i.e. when $\epsilon(j)=1$ for some $j\in[1,n+1]$), to
\begin{align*}
{\bf p}_{k(\varepsilon)}A^{\varepsilon(1)}(f_1) A^{\varepsilon(2)}(f_2) \ldots A^{\varepsilon(n)} (f_n)
\end{align*}
with
\[k(\varepsilon):=\min\{k_{1},\,k_h+\varepsilon(1)+
\ldots +\varepsilon(h-1):\ h\in[2,n+1]\text{ with } \epsilon(h)=1\} 
\]
Thus, establishing the normally ordered form of the expression in \eqref{QuanAlg02a} is equivalent to determining the normally ordered form of the expression in \eqref{QuanAlg02d}. Furthermore, we can always assume that $\varepsilon(1)=-1$ and $\varepsilon(n)=1$. In other words, the primary task is to determine the normally ordered form of a word in the following set:
\begin{align}\label{QuanAlg02e}
{\mathcal W}_{q,m}:=\{A^{\varepsilon(1)}(f_1)& A^{\varepsilon(2)}(f_2) \ldots A^{\varepsilon(n)} (f_n):n\in{\mathbb N}^*,\ \{f_1,\ldots,f_n\} \subset{\mathcal H}\notag\\
&\text{ and }\varepsilon\in\{-1,1\}^n\text{ with }\varepsilon(1)=-1\text{ and }\varepsilon(n)=1\}
\end{align} 
 
In the following, 
for any ${\bf f}:=\{f_k:\,k\in{\mathbb N}\}\subset{\mathcal H}$, for any $n\in {\mathbb N}$ and finite subset $I=\{i_1,\ldots, i_n\}\subset {\mathbb N}$ with the order $i_1<\ldots< i_n$, we introduce the following short notations:
\begin{align}\label{QuanAlg02e1}
A_{\bf f}^+ (I)&:=A^+(f_{i_1},\ldots ,f_{i_k})
:=A^+(f_{i_1}) \ldots A^+(f_{i_k}) \notag\\
A_{\bf f}(I)&:=A(f_{i_1}, \ldots ,f_{i_k}) :=A(f_{i_1}) \ldots A (f_{i_k})
\end{align}
It is worth noting that $A_{\bf f}^+ (I)$ may differ from the conjugate of $A_{\bf f}(I)$ if $\vert I\vert\ge 2$. Additionally, both $A_{\bf f}^+ (\emptyset)$ and  $A_{\bf f}(\emptyset)$ (specifically, $A_{\bf f}^+ (\{i_1,\ldots, i_n\})$ and $A_{\bf f}(\{i_1,\ldots, i_n\})$ with $n=0$) should be interpreted as the identity element of ${\mathcal Q}_{q,m}$.

As a first step, we prove the following result: 

\begin{proposition}\label{QuanAlg03} For any $n\in {\mathbb N}$ and $\{f,g_1,\ldots,g_n\}\subset{\mathcal H}$, utilizing the notations introduced in \eqref{QuanAlg02z} and \eqref{QuanAlg02e1}, we obtain:
\begin{align}\label{QuanAlg03a}
&A(f)A^+(g_1,\ldots,g_n)\notag\\
=&\big(q {\bf p}_{m}\big)^n A^+_{\bf g}((0,n])A(f)+\sum_{k=1}^n \big(q{\bf p}_{m}\big)^ {k-1}\langle f,g_k\rangle  A^+_{\bf g}((0,n]\setminus\{k\})
\end{align}  
\end{proposition}    

\begin{proof} For $n=0$, \eqref{QuanAlg03a} simplifies to the trivial equality $A(f)=A(f)$ because  $(q {\bf p}_m)^0$ equals to the identity, and as usual, $\sum_{k=1}^0:=0$.

When $n=1$, \eqref{QuanAlg03a}  reduces to \eqref{QuanAlg01c}. Assuming that \eqref{QuanAlg03a} holds for $n=N\ge1$, i.e.,
\begin{align}\label{QuanAlg03k}
&A(f)A^+_{\bf g}((0,N])\notag\\
=&\big(q {\bf p}_{m}\big)^NA^+_{\bf g}((0,N])A(f)
+\sum_{k=1}^N \big(q{\bf p}_{m}\big)^ {k-1}\langle f, g_k\rangle A^+_{\bf g}((0,N]\setminus\{k\})
\end{align}  
let's consider the case $n=N+1$. Using the fact that $A^+(g_1,\ldots,g_{N+1})=A^+_{\bf g}((0,N+1])=A^+_{\bf g}((0,N]) A^+(g_{N+1})$, \eqref{QuanAlg03k} trivially implies the following:
\begin{align}\label{QuanAlg03m}
&A(f)A^+(g_1,\ldots,g_{N+1})=A(f)A^+_{\bf g}((0,N]) A^+(g_{N+1})\notag\\
=&\big(q {\bf p}_{m}\big)^NA^+_{\bf g}((0,N])A(f) A^+(g_{N+1})\notag\\
&+\sum_{k=1}^N \big(q{\bf p}_{m}\big)^ {k-1}\langle f,g_k \rangle A^+_{\bf g}((0,N] \setminus\{k\}) A^+(g_{N+1})
\end{align}
The first $(q,m)$--commutation relation  \eqref{QuanAlg01c} yields $ A(f) A^+(g_{N+1})=\langle f,g_{N+1} \rangle +q{\bf p}_{m}A^+(g_{N+1})A(f)$, we can see that the expression in \eqref{QuanAlg03m} is equal to
\begin{align*}
&\langle f,g_{N+1}\rangle \big(q{\bf p}_m\big)^N A^+_{\bf g}((0,N])+\big(q{\bf p}_m\big)^{N+1} A^+_{\bf g}((0,N])A^+(g_{N+1})A(f)\\
+&\sum_{k=1}^N \big(q{\bf p}_{m}\big)^ {k-1}\langle f,g_k \rangle A^+_{\bf g}((0,N] \setminus\{k\}) A^+(g_{N+1})
\end{align*}
This expression can be rewritten as:
\begin{align*}
\big(q{\bf p}_m\big)^{N+1}A^+_{\bf g}((0,N+1]) A(f)
+\sum_{k=1}^{N+1} \big(q{\bf p}_{m}\big)^ {k-1}\langle f,g_k \rangle A^+_{\bf g}((0,N+1] \setminus\{k\}) 
\end{align*}
because 

$\bullet$ $\langle f,g_{N+1}\rangle \big(q{\bf p}_m\big)^NA^+_{\bf g}((0,N])=\langle f,g_k \rangle\big(q{\bf p}_{m}\big)^ {k-1} A^+_{\bf g}((0,N+1] \setminus\{k\}) \Big\vert_{k=N+1}$;

$\bullet$ $A^+_{\bf g}((0,N]\setminus\{k\})A^+(g_{N+1})=A^+_{\bf g}((0,N+1]\setminus\{k\})$ for any $k\in(0,N]$;

$\bullet$ $A^+_{\bf g}((0,N])A^+(g_{N+1})A(f)=A^+_{\bf g}((0,N+1])A(f)$.
 \end{proof} 

\begin{remark} Using the notations introduced prior to Proposition \ref{QuanAlg03} and defining \begin{align}\label{QuanAlg03a2}
{\mathbb N}_f&:=\text{the set of all finite subsets of } {\mathbb N}\notag\\
I(i)&:=I\cap (0,i),\  \text{ for any } I\in{\mathbb N}_f \text{ and } i\in {\mathbb N}
\end{align}
we can reinterpret \eqref{QuanAlg03a} as follows: for any ${\bf g}:=\{g_k:\,k\in {\mathbb N}\}\subset{\mathcal H}$ and $I\in{\mathbb N}_f$,
\begin{align}\label{QuanAlg03a0z}
A(f)A^+_{\bf g}(I)=\big(q {\bf p}_{m}\big)^{\vert I\vert}A^+_{\bf g}(I)A(f)+\sum_{i\in I} \big(q{\bf p}_{m}\big)^ {\vert I(i)\vert}\langle f,g_i\rangle A^+_{\bf g}(I\setminus \{i\})
\end{align}                        \end{remark}

Now, let's examine the normally ordered form of a generic element from the set ${\mathcal W}_{q,m}$, which is a word in the form:
\begin{align}\label{QuanAlg03a1}
A(f_1)A^+_{\bf g}(I_1)A(f_2)A^+_{\bf g}(I_2)\ldots A(f_n)A^+_{\bf g}(I_n)
\end{align}
here 

$\bullet$ $n\in{\mathbb N}^*$ and ${\bf g}:= \{g_k:\,k\in {\mathbb N}\}\subset{\mathcal H}$;

$\bullet$ $I_k\in {\mathbb N}_f$ for all $k\in\{1,\ldots,n\}$ and the sets $I_k$'s are {\bf completely ordered}, i.e., they are pairwise disjoint and $\max I_k<\min I_{k+1}$ holds for any $k$;

Moreover, a word in the form \eqref{QuanAlg03a1} does not imply that all annihilation operators are separated one by one by creation operators. In fact, some $I_k$ could be empty, and in this case, the definition specifies that $A(f_k)A^+_{\bf g}(I_k)A(f_{k+1})= A(f_k)A(f_{k+1})$. 

To represent the normally ordered form of the expression in \eqref{QuanAlg03a1} effectively, we'll introduce the following notations:

1) For any $n\ge 2$ and completely ordered sets $I_1,\ldots,I_n\in {\mathbb N}_f$,  
\begin{align}\label{QuanAlg04x}
I_k^h:=\cup_{j=k}^hI_j\text{ and } I_k^h(i):=I_k^h\cap (0,i)\quad \text{for any }  1\le k\le h\le n\text{ and }i\in {\mathbb N}
\end{align}

2) For any $n\ge 2$ and $k\in\{1,\ldots,n\}$, for any $n_0\in{\mathbb N}$ and completely ordered sets $I_{n_0+1},\ldots,I_{n_0+n}\in {\mathbb N}_f$, 
\begin{align}\label{QuanAlg06b1}
&[I_{n_0+1},\ldots,I_{n_0+n}]_k\notag\\
:=&\big\{ \{(l_h,r_h)\}_{h=1}^k :n_0+1\le l_1<\ldots <l_k\le n_0+n,\,\ r_k\in I_{l_k}^{n_0+n},\notag\\ 
&\hspace{2.4cm} r_{k-1}\in I_{l_{k-1}}^{n_0+n}\setminus \{r_k\}, \ldots,r_1\in I_{l_1}^{n_0+n} \setminus\{r_h:1< h\le k\} \big\}\notag\\
=&\big\{ \{(l_h,r_h)\}_{h=1}^k :n_0+1\le l_1<\ldots <l_k\le n_0+n \text{ and }\notag\\ 
&\hspace{2.4cm}r_h\in I_{l_h}^{n_0+n}\text{ for all } h\in(0,k]\text{ with }\vert\{r_1,\ldots,r_k\}\vert=k \big\}
\end{align}

3) For any $\{(l_h,r_h)\}_{h=1}^k\in[I_{n_0+1}, \ldots,I_{n_0+n}]_k$ and $s\in(n_0,n_0+n]$, by adopting the convention $l_0:=n_0$ and $l_{k+1}:=n_0+n+1$,
\begin{align} \label{QuanAlg06b}
&c_{I_{n_0+1},\ldots,I_{n_0+n}} (\{(l_h,r_h)\}_{h=1}^k; s)\notag\\
:=&\begin{cases}  
\big\vert I_s^{n_0+n}(r_p) \setminus \{r_{p+1},\ldots, r_k\} \big\vert ,&\text{if }s=l_p\text{ with }p\in \{1,\ldots,k\}\\ 
\big\vert I_s^{n_0+n} \setminus\{r_{p+1},\ldots, r_k\} \big\vert,&\text{if }l_p<s<l_{p+1} \text{ with }p\in \{0,\ldots,k\}
\end{cases}
\end{align}
It's evident that $c_{I_{n_0+1},\ldots, I_{n_0+n}} (\{(l_h, r_h)\}_{h=1}^{k};\cdot)$ is a function with values in ${\mathbb N}$ defined on the interval $(n_0,n_0+n]$. In this context, we interpret the set $\{r_{p+1},\ldots, r_k\}$ as empty when $p=k$, and additionally, we'll adopt the following conventions for the second case on the right--hand side of \eqref{QuanAlg06b}:

$\bullet$ as $l_{k+1}$ is defined as $n_0+n+1$, when $p=k$, the inequality $l_p<s<l_{p+1}$ simplifies to $l_k<s\le n_0+n$; and so
\[c_{I_{n_0+1},\ldots, I_{n_0+n}} (\{(l_h,r_h)\}_{h=1}^k; s) :=\big\vert I_s^{n_0+n} \setminus\{r_{k+1},\ldots, r_k\} \big\vert=\big\vert I_s^{n_0+n}\big\vert,\quad \forall s>l_k\]

$\bullet$ as $l_0$ is defined as $n_0$, when $p=0$, the inequality $l_p<s<l_{p+1}$ simplifies to $n_0<s<l_1$.

The case $n_0=0$ is particularly important. In this case,
\begin{align}\label{QuanAlg06b2}
[I_1,\ldots,I_n]_k&:=\big\{\{(l_h,r_h)\}_{h=1}^k : 1\le l_1<\ldots <l_k\le n \text{ and }\notag\\ 
&\hspace{0.8cm}r_h\in I_{l_h}^{n}\text{ for all }h\in (0,k]\text{ with }\vert\{r_1,\ldots,r_k\}\vert=k\big\}
\end{align}
and for any $\{(l_h,r_h)\}_{h=1}^k \in[I_1,\ldots,I_n]_k$, the function $c_{I_1,\ldots, I_n}(\{(l_h, r_h)\}_{h=1}^k;\cdot):(0,n]\longmapsto {\mathbb N}$ is defined as
\begin{align} \label{QuanAlg06b3}
&c_{I_1,\ldots, I_n} (\{(l_h,r_h)\}_{h=1}^k; s)\notag\\
:=&\begin{cases}  
\big\vert I_s^n(r_p)\setminus\{r_{p+1},\ldots,r_k\}\big \vert,&\text{if }s=l_p\text{ with }p\in\{1,\ldots,k\}\\ 
\big\vert I_s^n \setminus\{r_{p+1},\ldots, r_k\} \big\vert,&\text{if }l_p<s<l_{p+1} \text{ with }p\in \{0,\ldots,k\}
\end{cases}
\end{align}

\begin{theorem}\label{QuanAlg06}For any $n_0\in{\mathbb N}$, $ n\in{\mathbb N}^*$ and completely ordered sets $I_{n_0+1},\ldots,$ $I_{n_0+n} \in{\mathbb N}_f$, for any ${\bf g}:= \{g_k:\,k\in {\mathbb N}\}\subset{\mathcal H}$ and ${\bf f}:= \{f_k:\,k\in {\mathbb N}\} \subset {\mathcal H}$, following the notations and conventions mentioned above, we have the following normally ordered formula:
\begin{align}\label{QuanAlg06a}
&A(f_{n_0+1})A^+_{\bf g}(I_{n_0+1}) A(f_{n_0+2})A^+_{\bf g} (I_{n_0+2})\ldots A(f_{n_0+n})A^+_{\bf g}(I_{n_0+n})\notag\\
=&\Big(\prod_{s\in(0,n]}\big(q{\bf p}_{m+\vert I_{n_0+1}^{n_0+s-1}\vert-s+1}\big)^{\vert I_{n_0+s}^{n_0+n}\vert}\Big) \cdot A^+_{\bf g} (I_{n_0+1}^{n_0+n})\cdot A_{\bf f}((n_0,n_0+n])\notag\\
+&\sum_{k\in (0,n]}\sum_{\{(l_h,r_h)\}_{h=1}^k\in [I_{n_0+1},\ldots,I_{n_0+n}]_k}\prod_{p\in(0,k]} \langle f_{l_p},g_{r_p} \rangle\notag\\ 
&\ \cdot\Big(\prod_{s\in(0,n]} \big(q{\bf p}_{m+\vert I_{n_0+1}^{n_0+s-1} \vert-s+1} \big) ^{c_{I_{n_0+1},\ldots, I_{n_0+n}}(\{(l_h,r_h)\} _{h=1}^k;s)}\Big)\notag\\ 
& \ \cdot A^+_{\bf g} (I_{n_0+1}^{n_0+n} \setminus \{r_1,\ldots,r_k\}) \cdot A_{\bf f} ((n_0,n_0+n]\setminus\{l_1,\ldots,l_k\})
\end{align}
\end{theorem}   \smallskip

\begin{corollary}\label{QuanAlg06a9}
For any $n\in {\mathbb N}^*$ and completely ordered sets $I_1\ldots, I_n\in{\mathbb N}_f$, for any ${\bf g}:= \{g_k:\,k\in {\mathbb N}\}\subset{\mathcal H}$ and ${\bf f}:= \{f_k:\,k\in {\mathbb N}\}\subset{\mathcal H}$, we have the following normally ordered formula:
\begin{align}\label{QuanAlg06ax}
&A(f_1)A^+_{\bf g} (I_1) A(f_2)A^+_{\bf g}(I_2)\ldots A(f_n)A^+_{\bf g}(I_n)\notag\\
=&\Big(\prod_{s\in(0,n]}\big(q{\bf p}_{m+\vert I_{1}^{s-1}\vert-s+1}\big)^{\vert I_{s}^n\vert}\Big) \cdot A^+_{\bf g}(I_1^n)\cdot A_{\bf f}((0,n])\notag\\
+&\sum_{k\in (0,n]}\sum_{\{(l_h,r_h)\}_{h=1}^k\in [I_1,\ldots,I_n]_k}\prod_{h\in(0,k]} \langle f_{l_h}, g_{r_h}\rangle\notag\\ 
&\ \cdot\Big(\prod_{s\in(0,n]} \big(q{\bf p}_{m+\vert I_{1}^{s-1}\vert-s+1}\big)^{c_{I_1,\ldots, I_n}(\{(l_h,r_h)\}_{h=1}^k;s)}\Big)\notag\\ 
&\ \cdot A^+_{\bf g} (I_1^n\setminus \{r_1,\ldots,r_k\})\cdot A_{\bf f} ((0,n]\setminus\{l_1,\ldots,l_k\})
\end{align}
\end{corollary}\smallskip

It's evident that Corollary \ref{QuanAlg06a9} represents a specific case of Theorem \ref{QuanAlg06}. However, the following result establishes their equivalence. Therefore, our focus shifts to proving this Corollary rather than Theorem \ref{QuanAlg06}.

\begin{proposition}\label{QuanAlg09}Theorem \ref{QuanAlg06} and Corollary \ref{QuanAlg06a9} are equivalent.
\end{proposition} 
\begin{proof} By setting $n_0:=0$ in Theorem \ref{QuanAlg06}, we can deduce Corollary \ref{QuanAlg06a9}. Once we have established that Corollary, for any $n_0\in{\mathbb N}$, $ n\in{\mathbb N}^*$ and completely ordered sets $I_{n_0+1}\ldots, I_{n_0+n} \in{\mathbb N}_f$, along with the families ${\bf g}:= \{g_k:\,k\in {\mathbb N}\}\subset{\mathcal H}$ and ${\bf f}:= \{f_k:\,k\in {\mathbb N}\}\subset{\mathcal H}$, we introduce the following notations:
\begin{align}\label{QuanAlg09a}
J_h:=I_{n_0+h},\quad f'_h:=f_{n_0+h},\qquad \forall h\in(0,n]
\end{align}
These $J_h$'s are indeed completed ordered with the provided definitions, and we can express the left--hand side of \eqref{QuanAlg06a} as 
\begin{align}\label{QuanAlg069b}
A(f'_1)A^+_{\bf g}(J_1) A(f'_2)A^+_{\bf g} (J_2)\ldots A(f'_n)A^+_{\bf g}(J_n)
\end{align}
which is equal to, as said in Corollary \ref{QuanAlg06a9}, 

\begin{align}\label{QuanAlg09b}
&\Big(\prod_{s\in(0,n]}\big(q{\bf p}_{m+\vert J_{1}^{s-1}\vert-s+1}\big)^{\vert J_{s}^n\vert}\Big) \cdot A^+_{\bf g}(J_1^n)\cdot A_{\bf f'}((0,n])\notag\\
+&\sum_{k\in (0,n]}\sum_{\{(l'_h,r_h)\}_{h=1}^k \in [J_1,\ldots,J_n]_k}\prod_{h\in(0,k]} \langle f'_{l'_h}, g_{r_h}\rangle\notag\\ 
&\ \cdot\Big(\prod_{s\in(0,n]} \big(q{\bf p}_{m+\vert J_{1}^{s-1}\vert-s+1}\big) ^{c_{J_1,\ldots ,J_n}(\{(l'_h,r_h)\}_{h=1}^k;s)}\Big) \notag\\ 
&\ \cdot A^+_{\bf g} (J_1^n\setminus \{r_1,\ldots,r_k\})\cdot A_{\bf f'} ((0,n]\setminus\{l'_1,\ldots,l'_k\})
\end{align}
This expression is identical to the one on the right--hand side of \eqref{QuanAlg06a}. To see this evidence, we introduce $l_h:=l'_h+n_0$ for any $k\in(0,n]$ and $h\in (0,k]$, and conclude the following:

$\bullet$ $A^+_{\bf g}(J_1^n)=A^+_{\bf g}(I_{n_0+1}^ {n_0+n})$,\quad $A_{\bf f'}((0,n])=A_{\bf f}((n_0,n_0+n])$;

$\bullet$ $A^+_{\bf g}(J_1^n\setminus\{r_1, \ldots,r_k\}) =A^+_{\bf g} (I_{n_0+1}^{n_0+n}\setminus \{r_1,\ldots, r_k\})$ and $A_{\bf f'} ((0,n]\setminus\{l'_1,\ldots,l'_k\})=A_{\bf f} ((n_0,n_0+n]\setminus\{l_1,\ldots,l_k\})$';

$\bullet$ $\big(q{\bf p}_{m+\vert J_{1}^{s-1}\vert-s+1} \big)^{\vert J_{s}^n\vert}=\big(q{\bf p}_{m+\vert I_{n_0+1}^{n_0+s-1}\vert-s+1}\big)^{\vert I_{n_0+s}^{n_0+n}\vert}$ for any $s\in(0,n]$;

$\bullet$ $f'_{l'_h}=f_{l_h}$ and $\{(l_h,r_h)\}_{h=1}^k$ runs over $[I_{n_0+1},\ldots, I_{n_0+n}]_k$ as $\{(l'_h,r_h)\}_{h=1}^k$ varying over $[J_1,\ldots,J_n]_k$;

$\bullet$ for any $k,s\in(0,n]$ and for any $\{(l'_h, r_h)\} _{h=1}^k\in [J_1,\ldots,J_n]_k$ (equivalently, $\{(l_h,r_h)\}_{h=1}^k\in[I_{n_0+1},\ldots,I_{n_0+n}]_k $), \[{\bf p}_{m+\vert J_{1}^{s-1}\vert-s+1}={\bf p}_ {m+\vert I_{n_0+1}^{n_0+s-1} \vert-s+1}\]
and
\[c_{J_1,\ldots, J_n}(\{(l'_h,r_h)\}_{h=1}^k;s)=
c_{I_{n_0+1},\ldots, I_{n_0+n}}(\{(l_h,r_h)\}_{h=1}^k;s)
\]                   \end{proof}

\begin{proof} {\bf (of Corollary \ref{QuanAlg06a9})} For $n=1$, the condition $k\in(0, 1]$ means that $k=1$. Consequently, the condition $0 <l_1<\ldots <l_k\le 1$ is nothing else but $l_1=1$. So

$\bullet$ $[I_{1},\ldots,I_{n}]_1=[I_{1}]_1
=\{\{(1,r)\}:r\in I_{1}\}$ and $\{(1,r)\}$ runs over $[I_{1}]_1$ means that $r$ varies over $I_{1}$;

$\bullet$ $c_{I_{1},\ldots, I_{n}}(\{(l_h,r_h)\} _{h=1}^k; s)=c_{I_{1}}(\{(1,r)\}; s):=I_{1}(r)$.\\
Therefore, \eqref{QuanAlg06ax} becomes to \eqref{QuanAlg03a0z}:
\begin{align*}
A(f_{1})A^+_{\bf g} (I_{1}) 
=\big(q{\bf p}_{m}\big)^{\vert I_{1}\vert} A^+_{\bf g}(I_{1})\cdot A(f_{1})+\sum_{i\in I_{1}}\big(q{\bf p}_{m}\big)^{\vert I_{1}(i)\vert}\langle f_{1},g_{i}\rangle A^+_{\bf g} (I_{1}\setminus \{i\})
\end{align*}

Assuming the validity of \eqref{QuanAlg06ax} for $n=N$, let's extend it to $n=N+1$. Starting from \eqref{QuanAlg03a0z}, we obtain
\begin{align*}
A(f_{N+1})A^+_{\bf g} (I_{N+1}) 
=&\big(q{\bf p}_{m}\big)^{\vert I_{N+1}\vert} A^+_{\bf g}(I_{N+1})\cdot A(f_{N+1})\notag\\ 
+&\sum_{r\in I_{N+1}}\big(q{\bf p}_{m}\big)^{\vert I_{N+1}(r)\vert}\langle f_{N+1},g_r\rangle A^+_{\bf g} (I_{N+1}\setminus \{r\})
\end{align*}
and this leads to, in light of \eqref{QuanAlg01b}, 

\begin{align}\label{QuanAlg07f}
&A(f_{1})A^+_{\bf g} (I_{1})\ldots A(f_{N})A^+_{\bf g} (I_{N})A(f_{N+1})A^+_{\bf g} (I_{N+1}) \notag\\ 
=&\big(q{\bf p}_{m+\vert I_1^N\vert -N}\big)^{\vert I_{N+1}\vert} A(f_{1})A^+_{\bf g} (I_{1})\ldots A(f_{N-1})A^+_{\bf g} (I_{N-1}) A(f_{N})A^+_{\bf g} (I_{N}^{N+1}) A(f_{N+1})\notag\\ 
+&\sum_{r\in I_{N+1}}\langle f_{N+1},g_r\rangle
\big(q{\bf p}_{m+\vert I_1^N\vert -N}\big)^{\vert I_{N+1}(r)\vert} \notag\\
&\ A(f_{1})A^+_{\bf g} (I_{1})\ldots
A(f_{N-1})A^+_{\bf g} (I_{N-1}) A(f_{N})A^+_{\bf g} (I_{N}^{N+1}\setminus \{r\})
\end{align}
Now, we will separately analyse the product of creation and annihilation operators which arise in the two terms on the right--hand side of \eqref{QuanAlg07f}.

We introduce the following notations:
\begin{align}\label{QuanAlg07x1}
J_N:=I_{N}^{N+1},\quad &J_k:=I_{k} \text{ and } J_s^t:=\cup_{h=s}^t J_h\notag\\
&\forall 0<s\le t\le N \text{ and }k\in(0,N)
\end{align}
The sets $J_k$'s are obviously completely ordered due to the completely ordered property of the sets $I_k$'s. Furthermore, the induction's assumption yields 
\begin{align}\label{QuanAlg07g}
&A(f_{1})A^+_{\bf g} (I_{1})\ldots A(f_{N-1})A^+_{\bf g} (I_{N-1}) A(f_{N})A^+_{\bf g} (I_{N}^{N+1})A(f_{N+1})\notag\\ 
=&A(f_1)A^+_{\bf g} (J_1) \ldots A(f_N)A^+_{\bf g}(J_N)A(f_{N+1})\notag\\
=&\Big(\prod_{s\in(0,N]}\big(q{\bf p}_{m+\vert J_{1}^{s-1}\vert-s+1}\big)^{\vert J_{s}^N\vert}\Big) \cdot A^+_{\bf g}(J_1^N)\cdot A_{\bf f} ((0,N+1])\notag\\
+&\sum_{k\in (0,N]}\sum_{\{(l_h,r_h)\}_{h=1}^k\in [J_1,\ldots,J_N]_k}\Big(\prod_{s\in(0,N]} \big(q{\bf p}_{m+\vert J_{1}^{s-1}\vert-s+1}\big)^{c_{J_1,\ldots, J_N}(\{(l_h,r_h)\}_{h=1}^k;s)}\Big)\notag\\ 
&\ \cdot\prod_{h\in(0,k]} \langle f_{l_h},g_{r_h}\rangle\cdot A^+_{\bf g} (J_1^N\setminus \{r_1,\ldots,r_k\})\cdot A_{\bf f} ((0,N+1]\setminus\{l_1,\ldots,l_k\})
\end{align}
Likewise, by introducing the following notations for any given $r\in I_{N+1}$:
\begin{align}\label{QuanAlg07x2}
H_N:=I_{N}^{N+1}\setminus\{r\},\quad &H_k:=I_{k}\ 
\text{ and }\ H_s^t:=\cup_{h=s}^t H_h\notag\\
&\forall 0<s\le t\le N \text{ and }k\in(0,N)
\end{align}
the completely ordered property of the sets $I_k$'s ensures that the sets $H_k$'s are also completely ordered. Furthermore, the induction's assumption provides us
\begin{align}\label{QuanAlg07g0}
&A(f_{1})A^+_{\bf g} (I_{1})\ldots A(f_{N-1})A^+_{\bf g} (I_{N-1}) A(f_{N})A^+_{\bf g} (I_{N}^{N+1} \setminus\{r\})\notag\\
=&A(f_1)A^+_{\bf g} (H_1) \ldots A(f_N)A^+_{\bf g}(H_N)\notag\\
=&\Big(\prod_{s\in(0,N]}\big(q{\bf p}_{m+\vert H_{1}^{s-1}\vert-s+1}\big)^{\vert H_{s}^N\vert}\Big) \cdot A^+_{\bf g}(H_1^N)\cdot A_{\bf f}((0,N])\notag\\
+&\sum_{k\in (0,N]}\sum_{\{(l_h,r_h)\}_{h=1}^k\in [H_1,\ldots,H_N]_k}\Big(\prod_{s\in(0,N]} \big(q{\bf p}_{m+\vert H_{1}^{s-1}\vert-s+1}\big)^{c_{H_1,\ldots, H_N}(\{(l_h,r_h)\}_{h=1}^k;s)}\Big)\notag\\ 
&\ \cdot\prod_{h\in(0,k]} \langle f_{l_h},g_{r_h}\rangle\cdot A^+_{\bf g} (H_1^N \setminus \{r_1,\ldots,r_k\})\cdot A_{\bf f} ((0,N]\setminus\{l_1,\ldots,l_k\})
\end{align}
By applying these results, we can express the expression in \eqref{QuanAlg07f} as a sum of four terms:
\begin{align*}
&\big(q{\bf p}_{m+\vert I_1^N\vert -N}\big)^{\vert I_{N+1}\vert} \Big(\prod_{s\in(0,N]}\big(q{\bf p}_{m+\vert J_{1}^{s-1}\vert-s+1}\big)^{\vert J_{s}^N\vert}\Big) \cdot A^+_{\bf g}(J_1^N)\cdot A_{\bf f} ((0,N+1])\notag\\
+&\sum_{k\in (0,N]}\sum_{\{(l_h,r_h)\}_{h=1}^k\in [J_1,\ldots,J_N]_k}\prod_{h\in(0,k]} \langle f_{l_h},g_{r_h}\rangle\cdot\big(q{\bf p}_{m+\vert I_1^N\vert -N}\big)^{\vert I_{N+1}\vert}  \notag\\ 
& \ \cdot \Big(\prod_{s\in(0,N]} \big(q{\bf p}_{m+\vert J_{1}^{s-1}\vert-s+1}\big) ^{c_{ J_1,\ldots,J_N}(\{(l_h,r_h)\}_{h=1}^k;s)}\Big)\notag\\ 
& \ \cdot A^+_{\bf g} (J_1^N\setminus \{r_1,\ldots,r_k\}) \cdot A_{\bf f} ((0,N+1]\setminus\{l_1,\ldots,l_k\})
\end{align*}
\begin{align}\label{QuanAlg07g1}
+&\sum_{r\in I_{N+1}}\langle f_{N+1},g_r\rangle \cdot
\big(q{\bf p}_{m+\vert I_1^N\vert -N}\big)^{\vert I_{N+1}(r)\vert} \notag\\
&\ \cdot\Big(\prod_{s\in(0,N]} \big(q{\bf p}_{m+\vert H_{1}^{s-1}\vert-s+1}\big) ^{\vert H_{s}^N\vert}\Big)\cdot A^+_{\bf g} (H_1^N)\cdot A_{\bf f} ((0,N]) \notag\\
+&\sum_{r\in I_{N+1}}\sum_{k\in(0,N]} \sum_{\{(l_h,r_h)\}_{h=1}^k\in [H_1,\ldots,H_N]_k}
\langle f_{N+1},g_r\rangle \prod_{h\in(0,k]} \langle f_{l_h},g_{r_h}\rangle\notag\\ 
&\ \cdot \big(q{\bf p}_{m+\vert I_1^N\vert -N}\big)^{\vert I_{N+1}(r)\vert}\cdot
\Big(\prod_{s\in(0,N]} \big(q{\bf p}_{m+\vert H_{1}^{s-1}\vert-s+1}\big)^{c_{H_1,\ldots, H_N}(\{(l_h,r_h)\}_{h=1}^k;s)}\Big)\notag\\
&\ \cdot A^+_{\bf g} (H_1^N\setminus \{r_1,\ldots,r_k\})\cdot A_{\bf f} ((0,N]\setminus\{l_1,\ldots,l_k\})
\end{align}

As a first step, following the definition of the $J_k$'s given in \eqref{QuanAlg07x1}, we can identify that in \eqref{QuanAlg07g1}, the term lacking any scalar product of the form $\langle f_l,g_r\rangle$ corresponds to its initial expression, which can be expressed as:
\begin{align}\label{QuanAlg07h1}
\Big(\prod_{s\in(0,N+1]}\big(q{\bf p}_{m+\vert I_{1}^{s-1}\vert-s+1}\big)^{\vert I_{s}^{N+1}\vert}\Big) \cdot A^+_{\bf g}(I_1^{N+1})\cdot A_{\bf f} ((0,N+1])
\end{align}

As a second step, let's examine the terms in \eqref{QuanAlg07g1} that involve exactly one scalar product in the form $\langle f_l,g_r\rangle$, i.e.
\begin{align}\label{QuanAlg07h2}
&\sum_{\{(l_1,r_1)\}\in [J_1,\ldots,J_N]_1}
\langle f_{l_1},g_{r_1}\rangle \big(q{\bf p}_{m+\vert I_1^N\vert -N}\big)^{\vert I_{N+1}\vert}\notag\\ 
&\ \cdot \Big(\prod_{s\in(0,N]} \big(q{\bf p}_{m+\vert J_{1}^{s-1}\vert-s+1}\big) ^{c_{ J_1,\ldots, J_N}(\{(l_1,r_1)\};s)}\Big) \notag\\ 
&\ \cdot A^+_{\bf g} (J_1^N\setminus \{r_1\})\cdot A_{\bf f} ((0,N+1]\setminus\{l_1\})\notag\\
+&\sum_{r\in I_{N+1}} \langle f_{N+1},g_r\rangle \big(q{\bf p}_{m+\vert I_1^N\vert -N}\big)^{\vert I_{N+1}(r)\vert} \notag\\
&\ \cdot\Big(\prod_{s\in(0,N]} \big(q{\bf p}_{m+\vert H_{1}^{s-1}\vert-s+1}\big)^{ \vert H_{s}^N\vert}\Big)\cdot A^+_ {\bf g}(H_1^N)\cdot A_{\bf f}((0,N])
\end{align}
Substituting $n=N+1$ and $k=1$ in \eqref{QuanAlg06b}, we obtain
\[\{0,\ldots,k\}=\{0,1\},\qquad \{r_{p+1},\ldots, r_1\}=\begin{cases}\{r_1\},&\text{if }p=0\\ \emptyset,&\text{if }p=1\end{cases}\]
and this leads to the following:
\begin{align} \label{QuanAlg07h2a}
&c_{I_{1},\ldots, I_{N+1}} (\{(l_1,r_1)\}; s)\notag\\
=&\begin{cases}  
\big\vert I_{l_1}^{N+1}(r_1) \big\vert
,&\text{ if }s=l_1\\ 
\big\vert I_s^{N+1} \setminus\{r_{p+1},\ldots, r_1\} \big\vert,&\text{ if }l_p<s<l_{p+1} \text{ with }p\in \{0,1\}   \end{cases}\notag\\
=&\begin{cases}  
\big\vert I_{l_1}^{N+1}(r_1) \big\vert
,&\text{if }s=l_1\\ 
\big\vert I_s^{N+1} \setminus\{r_1\} \big\vert,
&\text{if }s<l_1\\ 
\big\vert I_s^{N+1} \big\vert,&\text{if }s>l_1
\end{cases}
\end{align}
On the other hand, the expression for $k=1$ in the second term on the right--hand side of \eqref{QuanAlg06ax} is
\begin{align} \label{QuanAlg07h2b}
&\sum_{\{(l,r)\}\in[I_1,\ldots,I_{N+1}]_1} \Big(\prod_{s\in (0,N+1]}\big(q{\bf p}_{m+\vert I_{1}^{s-1}\vert-s+1} \big)^{c_{I_1,\ldots, I_{N+1}}(\{(l,r)\};s)}\Big)\notag\\ 
& \hspace{2.7cm} \cdot \langle f_{l},g_{r}\rangle\cdot A^+_{\bf g} (I_1^{N+1}\setminus \{r\})\cdot A_{\bf f} ((0,N+1]\setminus\{l\})
\end{align} 
Applying the equality:
\[\sum_{\{(l,r)\}\in[I_1,\ldots,I_{N+1}]_1}=\sum_{l\in (0,N+1]}\sum_{r\in I_l^{N+1}}=\sum_{r\in I_{N+1}}+\sum_{l\in (0,N]}\sum_{r\in I_l^{N+1}}
\]
we conclude that the expression \eqref{QuanAlg07h2b} is equal to

\begin{align} \label{QuanAlg07h2c}
&\sum_{r\in I_{N+1}}\Big(\prod_{s\in (0,N+1]} \big(q{\bf p}_{m+\vert I_1^{s-1}\vert-s+1} \big)^ {c_{I_1,\ldots,I_{N+1}}(\{(N+1,r)\};s)}\Big)\notag\\ 
& \hspace{1.2cm}\cdot \langle f_{N+1},g_r\rangle\cdot A^+_{\bf g} (I_1^{N+1}\setminus \{r\})\cdot A_{\bf f} ((0,N]) \notag\\
+&\sum_{l\in (0,N]}\sum_{r\in I_l^{N+1}} \Big(\prod_{s \in(0,N+1]}\big(q{\bf p}_{m+\vert I_{1}^{s-1}\vert-s+1} \big)^{c_{I_1,\ldots,I_{N+1}}(\{ (l,r)\};s)}\Big)\notag\\ 
& \hspace{1.2cm}\cdot \langle f_l,g_r\rangle\cdot A^+_{\bf g} (I_1^{N+1}\setminus \{r\})\cdot A_{\bf f} ((0,N+1]\setminus\{l\})
\end{align} 
According to the definition, we have:
\begin{align*}
H_{s}^N\overset{\eqref{QuanAlg07x2}}=
I_{s}^{N+1}\setminus\{r\},\qquad \forall r\in I_{N+1},\ s\in(0,N]
\end{align*}
and when $l_1=N+1$, the expression in \eqref{QuanAlg07h2a} can be simplified as follows:
\begin{align*}
c_{I_1,\ldots,I_{N+1}} (\{(N+1, r_1)\};s)&=c_{I_1,\ldots, I_{N+1}}(\{(l_1,r_1)\};s)\\
&=\begin{cases}  
\big\vert I_{l}^{N+1}(r) \big\vert
,&\text{ if }s=N+1=l_1\\ 
\big\vert I_s^{N+1} \setminus\{r_1\} \big\vert,
&\text{ if }s\in(0,N]
\end{cases}
\end{align*}
Therefore, in the expression on the right--hand side of \eqref{QuanAlg07h2c}, $A^+_{\bf g} (I_1^{N+1}\setminus \{r\})$ is, in fact, $A^+_{\bf g} (H_1^{N})$. Moreover
\begin{align*} 
&\Big(\prod_{s\in (0,N+1]}\big(q{\bf p}_{m+\vert I_1^{s-1}\vert-s+1} \big)^{c_{I_1,\ldots, I_{N+1}}(\{(N+1,r)\};s)}\Big)\\
=&\big(q{\bf p}_{m+\vert I_1^{s-1}\vert-s+1} \big)^{c_{I_1,\ldots,I_{N+1}}(\{(N+1,r)\};s)}\Big\vert _{s=N+1}\\
&\cdot\Big(\prod_{s\in (0,N]}\big(q{\bf p}_{m+\vert I_1^{s-1}\vert-s+1} \big)^{c_{I_1,\ldots, I_{N+1}} (\{(N+1,r)\};s)}\Big)\\
=&\big(q{\bf p}_{m+\vert I_1^{N}\vert-N} \big)^{\vert I_{N+1}(r)\vert} \Big(\prod_{s\in (0,N]}\big(q{\bf p}_{m+\vert I_1^{s-1}\vert-s+1} \big)^{\big\vert I_s^{N+1} \setminus\{r\} \big\vert}\Big)\\
=&\big(q{\bf p}_{m+\vert I_1^{N}\vert-N} \big)^{\vert I_{N+1}(r)\vert} \Big(\prod_{s\in (0,N]}\big(q{\bf p}_{m+\vert I_1^{s-1}\vert-s+1} \big)^{\big\vert H_s^{N}\big\vert}\Big)
\end{align*} 
Applying this to \eqref{QuanAlg07h2c}, it becomes evident that the first term of the expression \eqref{QuanAlg07h2} is equal to the first term of the expression \eqref{QuanAlg07h2c}. We'll now proceed to demonstrate that this assertion remains valid when ``{\it first}'' is replaced with ``{\it second}''. In other words, we'll prove that
\begin{align}\label{QuanAlg07j}
&\sum_{\{(l_1,r_1)\}\in [J_1,\ldots,J_N]_1}\langle f_{l_1},g_{r_1}\rangle\cdot\big(q{\bf p}_{m+\vert I_1^N \vert -N}\big)^{\vert I_{N+1}\vert} \notag\\ 
& \hspace{2.3cm} \cdot \Big(\prod_{s\in(0,N]} \big(q{\bf p}_{m+\vert J_{1}^{s-1}\vert-s+1}\big) ^{c_{J_1,\ldots, J_N}(\{(l_1,r_1)\};s)}\Big)\notag\\ 
&\hspace{2.3cm}\cdot A^+_{\bf g}(J_1^N\setminus\{r_1\}) \cdot A_{\bf f} ((0,N+1]\setminus\{l_1\})\notag\\ 
=&\sum_{l\in (0,N]}\sum_{r\in I_l^{N+1}} \Big(\prod_{s \in(0,N+1]}\big(q{\bf p}_{m+\vert I_{1}^{s-1} \vert-s+1}\big)^{c_{I_1,\ldots,I_{N+1}}(\{(l,r)\};s)} \Big)\notag\\ 
& \hspace{2.3cm}\cdot \langle f_l,g_r\rangle\cdot A^+_{\bf g} (I_1^{N+1}\setminus \{r\})\cdot A_{\bf f} ((0,N+1]\setminus\{l\})
\end{align} 
In accordance with \eqref{QuanAlg07x1} which gives that $J_1^{N}=I_1^{N+1}$ and $J_1^{s-1}=I_1^{s-1}$ for all $s\in(0,N]$, the expression on the left--hand side of \eqref{QuanAlg07j} is equal to, 
\begin{align*}
&\sum_{l\in (0,N]}\sum_{r\in I_l^{N+1}}
\big(q{\bf p}_{m+\vert I_1^N\vert -N}\big)^{\vert I_{N+1}\vert} \Big(\prod_{s\in(0,N]} \big(q{\bf p}_{m+\vert I_{1}^{s-1}\vert-s+1}\big) ^{c_{J_1,\ldots,J_N} (\{(l,r)\};s)} \Big)\\ 
& \hspace{1.8cm}\cdot\langle f_{l},g_{r}\rangle \cdot A^+_{\bf g} (I_1^{N+1}\setminus \{r\})\cdot A_{\bf f} ((0,N+1]\setminus\{l\})
\end{align*} 
Hence, to confirm \eqref{QuanAlg07j}, it suffices to show that, for any $l\in(0,N]$ and $r\in I_l^{N+1}$, 
\begin{align}\label{QuanAlg07j1}
&\big(q{\bf p}_{m+\vert I_1^N\vert -N}\big)^{\vert I_{N+1}\vert} \Big(\prod_{s\in(0,N]} \big(q{\bf p}_{m+\vert I_{1}^{s-1}\vert-s+1}\big) ^{c_{J_1,\ldots,J_N} (\{(l,r)\};s)} \Big)\notag\\
=&\prod_{s \in(0,N+1]}\big(q{\bf p}_{m+\vert I_{1}^{s-1} \vert-s+1} \big)^{c_{I_1,\ldots, I_{N+1}}(\{(l,r)\};s)}
\end{align} 
By utilizing the fact that $l\le N<N+1$ and the formula \eqref{QuanAlg07h2a}, we can derive the following:
\begin{align*}
&\text{the expression on the right--hand side of \eqref{QuanAlg07j1}}\\
=&\big(q{\bf p}_{m+\vert I_{1}^N \vert-N} \big)^{c_{I_1,\ldots,I_{N+1}}(\{(l,r)\};N+1)}\prod_{s \in(0,N]}\big(q{\bf p}_{m+\vert I_{1}^{s-1} \vert-s+1} \big)^{c_{I_1,\ldots, I_{N+1}}(\{(l,r)\};s)}\\
=&\big(q{\bf p}_{m+\vert I_{1}^N \vert-N} \big)^{\vert I_{N+1}\vert} \prod_{s \in(0,N]}\big(q{\bf p}_{m+\vert I_{1}^{s-1} \vert-s+1} \big)^{c_{I_1,\ldots, I_{N+1}}(\{(l,r)\};s)}
\end{align*}
Moreover, \eqref{QuanAlg07j1} follows (therefore \eqref{QuanAlg07j}) because
\begin{align*}
c_{J_{1},\ldots, J_{N}} (\{(l,r)\}; s)
\overset{\eqref{QuanAlg07h2a}}
=&\begin{cases}  
\big\vert J_{l}^{N}(r) \big\vert
,&\text{if }s=l\\ 
\big\vert J_s^{N} \setminus\{r\} \big\vert,
&\text{if }s<l\\ 
\big\vert J_s^{N} \big\vert,&\text{if }s>l
\end{cases}= \begin{cases}  
\big\vert I_{l}^{N+1}(r) \big\vert
,&\text{if }s=l\\ 
\big\vert I_s^{N+1} \setminus\{r\} \big\vert,
&\text{if }s<l\\ 
\big\vert I_s^{N+1} \big\vert,&\text{if }s>l
\end{cases}\\
\overset{\eqref{QuanAlg07h2a}}=&
c_{I_{1},\ldots, I_{N+1}} (\{(l,r)\}; s)
\end{align*}

As a third step, let's examine the terms in \eqref{QuanAlg07g1} that contain exactly $N+1$ scalar products in the form of $\langle f_l,g_r\rangle$, i.e.
\begin{align}\label{QuanAlg07h0}
&\sum_{r\in I_{N+1}}\sum_{\{(l_h,r_h)\}_{h=1}^N\in [H_1,\ldots,H_N]_N} \langle f_{N+1},g_r\rangle \prod _{h\in(0,N]} \langle f_{l_h},g_{r_h}\rangle \notag\\ 
&\ \cdot \big(q{\bf p}_{m+\vert I_1^N\vert -N}\big)^{\vert I_{N+1}(r)\vert}\cdot
\Big(\prod_{s\in(0,N]} \big(q{\bf p}_{m+\vert H_{1}^{s-1}\vert-s+1}\big)^{c_{H_1,\ldots, H_N}(\{(l_h,r_h)\}_{h=1}^N;s)}\Big)\notag\\
&\ \cdot A^+_{\bf g} (H_1^N\setminus \{r_1,\ldots,r_N\})\cdot A_{\bf f} ((0,N]\setminus\{l_1,\ldots,l_N\})
\end{align}
The strictly increasing property of $l_h$'s implies the following: 

$\bullet$ $l_{N+1}=N+1$ and $\{l_1,\ldots,l_N\}=\{1,\ldots, N\}$ (consequently, $ A_{\bf f} ((0,N]\setminus\{l_1,\ldots,l_N\})=A_{\bf f}(\emptyset)={\bf 1}$ and $\langle f_{N+1},g_r\rangle \prod _{h\in(0,N]} \langle f_{l_h},g_{r_h}\rangle= \prod _{h\in(0,N+1]} \langle f_{h},g_{r_h}\rangle$);

$\bullet$ to perform the summation for all $\{(l_h,r_h)\}_{h=1}^N$ over $[H_1,\ldots,H_N]_N$ means that we sum all $r_p\in H^N_p$ with $p\in(0,n]$ where $\vert\{r_1,\ldots,r_N\}\vert= N$; more specifically, this involves summing all $r_N$ in $H_{N}^N$, all $r_{N-1}$ in $H_{N-1}^N$ while excluding $r_N$, and so forth, until we reach $r_1$ in $H_{1}^N$ excluding all $r_h$ for $1< h\le N$.\\
So the expression \eqref{QuanAlg07h0} can be rephrased, by substituting $r$ with $r_{N+1}$, as follows:
\begin{align}\label{QuanAlg07h}
&\sum_{r_{N+1}\in I_{N+1},r_N\in H_{N}^N,r_{N-1}\in H_{N-1}^N\setminus \{r_N\},\atop \ldots,r_1\in H_{1}^N \setminus\{r_h:1< h\le N\}}\prod_{h\in(0,N+1]} \langle f_{h},g_{r_h}\rangle\cdot \big(q{\bf p}_{m+\vert I_1^N \vert -N}\big)^{\vert I_{N+1}(r_{N+1})\vert}\notag\\ 
&\Big(\prod_{s\in(0,N]} \big(q{\bf p}_{m+\vert H_{1}^{s-1}\vert-s+1}\big)^{c_{H_1,\ldots, H_N}(\{(h,r_h)\}_{h=1}^N;s)}\Big)\cdot A^+_{\bf g} (H_1^N\setminus \{r_1,\ldots,r_N\})
\end{align}
Moreover, \eqref{QuanAlg07x2} says that
\begin{align*}
&H_1^{s-1}=I_1^{s-1} \text{ and } H_s^N=I_s^{N+1}
\setminus\{r_{N+1}\},\qquad \forall s\in(0,N]\\
&H_1^N\setminus\{r_1,\ldots,r_{N}\}=I_1^{N+1}\setminus\{r_1,\ldots,r_{N+1}\}
\end{align*}
Utilizing these results and noting that all intervals $(l_h,l_{h+1})$'s are empty when $l_k=k$ for all $k$, we can deduce
\begin{align*}
&c_{H_1,\ldots,H_N}(\{(h,r_h)\}_{h=1}^N;s) \notag\\
\overset{\eqref{QuanAlg06b3}} =&\big\vert H_s^ N(r_s)\setminus\{r_{s+1},\ldots,r_N\}\big\vert
=\big\vert I_s^{N+1} (r_s)\setminus\{r_{s+1},\ldots,
r_N,r_{N+1}\}\big\vert, \quad\forall s\in(0,N]
\end{align*}
and similarly
\begin{align*}
c_{I_1,\ldots,I_{N+1}} (\{(h,r_h)\}_{h=1}^{N+1};s)
=&\big\vert I_s^{N+1} (r_s)\setminus\{r_{s+1},\ldots,
r_{N+1}\}\big\vert\\
=&\begin{cases}  
\big\vert I_{N+1} (r_{N+1}),&\text{ if }s=N+1\\ 
\big\vert I_s^{N+1} (r_s)\setminus\{r_{s+1},\ldots,
r_{N+1}\}\big\vert,&\text{ if }s\in(0,N]
\end{cases}
\end{align*}
With the help of these results, we obtain
\begin{align}\label{QuanAlg07h3}
&\text{the expression \eqref{QuanAlg07h}}\notag\\
=&\sum_{r_{N+1}\in I_{N+1},r_{N}\in I_{N}^{N+1}\setminus \{r_{N+1}\},\atop \ldots,r_1\in I_{1}^{N+1} \setminus\{r_h:1< h\le N+1\}} \prod_{h \in(0,N+1]} \langle f_{h},g_{r_h}\rangle\notag\\ 
&\ \cdot\Big(\prod_{s\in(0,N+1]} \big(q{\bf p}_{m+\vert I_{1}^{s-1}\vert-s+1}\big)^{c_{I_1,\ldots, I_{N+1}}(\{(h,r_h)\}_{h=1}^{N+1};s)}\Big)\notag\\ 
&\cdot A^+_{\bf g} (I_1^{N+1}\setminus \{r_1,\ldots,r_{N+1}\})
\end{align}

As a fourth step, let's see the terms in \eqref{QuanAlg07g1} that involve $k$'s scalar products of the form $\langle f_l,g_r\rangle$ with $k\in(1,N]$, i.e.
\begin{align}\label{QuanAlg07k}
&\sum_{\{(l_h,r_h)\}_{h=1}^k\in [J_1,\ldots,J_N]_k} \prod_{h\in(0,k]} \langle f_{l_h},g_{r_h}\rangle\cdot\big(q{\bf p}_{m+\vert I_1^N\vert -N}\big)^{\vert I_{N+1}\vert}\notag\\ 
&\ \cdot \Big(\prod_{s\in(0,N]} \big(q{\bf p}_{m+\vert J_{1}^{s-1}\vert-s+1}\big) ^{c_{ J_1, \ldots,J_N}(\{(l_h,r_h)\}_{h=1}^k;s)}\Big) \notag\\ 
& \ \cdot A^+_{\bf g} (J_1^N\setminus \{r_1,\ldots,r_k\}) \cdot A_{\bf f} ((0,N+1]\setminus\{l_1,\ldots,l_k\})\notag\\ 
+&\sum_{r\in I_{N+1}}\sum_{\{(l_h,r_h)\}_{h=1}^{k-1}\in [H_1,\ldots,H_N]_{k-1}}\langle f_{N+1},g_r\rangle \prod _{h\in(0,k-1]}\langle f_{l_h},g_{r_h}\rangle \notag\\
&\ \cdot \big(q{\bf p}_{m+\vert I_1^N\vert -N}\big)^{\vert I_{N+1}(r)\vert}
\cdot \Big(\prod_{s\in(0,N]} \big(q{\bf p}_{m+\vert H_{1}^{s-1}\vert-s+1}\big)^{c_{H_1,\ldots, H_N}(\{(l_h,r_h)\}_{h=1}^{k-1};s)}\Big)\notag\\
&\ \cdot A^+_{\bf g} (H_1^N\setminus \{r_1,\ldots,r_{k-1}\})\cdot A_{\bf f} ((0,N]\setminus\{l_1,\ldots,l_{k-1}\})
\end{align}

Let's begin by examining the first term in \eqref{QuanAlg07k}. The definition of $J_k$'s provided in \eqref{QuanAlg07x1} ensures that
\[
J_s^N\setminus E=I_s^{N+1}\setminus E\ \text{ and }\ 
J_{1}^{s-1}=I_{1}^{s-1}\ \text{ for any }s\in(0,N]
\text{ and }E\subset{\mathbb N}\]
Therefore, to facilitate further analysis, we rewrite both the summation over $\{(l_h,r_h)\} _{h=1}^k$ within $[J_1,\ldots,J_N]_k$ and the expression $c_{ J_1,\ldots, J_N}(\{(l_h,r_h)\}_{h=1}^k;s)$ as follows:
\begin{align*}
&\sum_{\{(l_h,r_h)\}_{h=1}^k\in [J_1,\ldots,J_N]_k}
=\sum_{1\le l_1<\ldots l_k\le N}\sum_{r_k\in J_{l_k} ^N,r_{k-1}\in J_{l_{k-1}}^N\setminus\{r_k\},\atop \ldots,r_1\in J_{l_1}^N\setminus\{r_2,\ldots,r_k\}}\\
=&\sum_{1\le l_1<\ldots l_k\le N}\sum_{r_k\in I_{l_k}^{N+1}, r_{k-1}\in I_{l_{k-1}}^{N+1}\setminus \{r_k\},\atop\ldots, r_1\in I_{l_{1}}^{N+1}\setminus \{r_2,\ldots,r_k\}}=\sum_{\{(l_h,r_h)\}_{h=1}^k\in [I_1,\ldots,I_{N+1}]_k\atop l_k\le N}
\end{align*}
and
\begin{align*}
&c_{ J_1,\ldots, J_N}(\{(l_h,r_h)\}_{h=1}^k;s)\notag\\ \overset{\eqref{QuanAlg06b3}}
=&\begin{cases}  
\big\vert J_s^N(r_p)\setminus\{r_{p+1},\ldots,r_k\}\big \vert,&\text{ if }s=l_p\text{ with }p\in(0,k]\\ 
\big\vert J_s^N \setminus\{r_{p+1},\ldots, r_k\}\big \vert,&\text{ if }l_p<s<l_{p+1}\text{ with }p\in[0,k]
\end{cases}\\
=&\begin{cases}  
\big\vert I_s^{N+1}(r_p)\setminus\{r_{p+1},\ldots,r_k\} \big\vert,&\text{ if }s=l_p\text{ with }p\in (0,k] \\ 
\big\vert I_s^{N+1} \setminus\{r_{p+1},\ldots, r_k\} \big\vert,&\text{ if }l_p<s<l_{p+1} \text{ with }p\in [0,k]
\end{cases}\\
\overset{\eqref{QuanAlg06b3}}=&c_{ I_1,\ldots, I_{N+1}}(\{(l_h,r_h)\}_{h=1}^k;s) 
\end{align*}
where, it is worth noting that the fact $l_k\le N<N+1$ guarantees the following equalities:
\[
c_{ I_1,\ldots, I_{N+1}} (\{(l_h,r_h)\}_{h=1}^k;N+1)=\big\vert I_{N+1}^{N+1} \setminus\{r_{k+1},\ldots, r_k\} \big\vert
=\big\vert I_{N+1}\big\vert
\]
Summing up, we conclude that
\begin{align}\label{QuanAlg07k1}
&\text{the first expression in \eqref{QuanAlg07k}}\notag\\
=&\sum_{\{(l_h,r_h)\}_{h=1}^k\in [I_1,\ldots,I_{N+1}]_k\atop l_k\le N}
\Big(\prod_{s\in(0,N+1]} \big(q{\bf p}_{m+\vert I_{1}^{s-1}\vert-s+1}\big) ^{c_{I_1,\ldots, I_{N+1}} (\{(l_h,r_h)\}_{h=1}^k;s)}\Big) \notag\\ 
&\ \cdot\Big(\prod_{h\in(0,k]} \langle f_{l_h},g_{r_h}\rangle\Big)\cdot A^+_{\bf g} (I_1^{N+1}\setminus \{r_1,\ldots,r_k\}) \cdot A_{\bf f} ((0,N+1]\setminus\{l_1,\ldots,l_k\})
\end{align}

Now let's proceed to demonstrate the following equality:
\begin{align}\label{QuanAlg07n}
&\text{the second expression in \eqref{QuanAlg07k}}\notag\\
=&\sum_{\{(l_h,r_h)\}_{h=1}^k\in[I_1,\ldots,I_{N+1}]_k \atop l_k= N+1}\Big(\prod_{s\in(0,N+1]} \big(q{\bf p}_{m+\vert I_{1}^{s-1}\vert-s+1}\big)^{c_{I_1,\ldots, I_{N+1}} (\{(l_h,r_h)\}_{h=1}^k;s)}\Big) \notag\\ 
&\ \cdot \Big(\prod _{h\in(0,k]} \langle f_{l_h},g_{r_h}\rangle\Big)\cdot A^+_{\bf g} (I_1^{N+1}\setminus \{r_1,\ldots,r_k\}) \cdot A_{\bf f} ((0,N+1]\setminus\{l_1,\ldots,l_k\})
\end{align}

It is obvious that the condition $l_k= N+1$ implies the following equality:
\[(0,N+1]\setminus \{l_1,\ldots,l_k\}=(0,N]\setminus\{l_1,\ldots,l_{k-1}\}\] 
Furthermore, substituting $r$ in the second expression in \eqref{QuanAlg07k} with $r_k$ and using the definition of $H_k$'s as given in \eqref{QuanAlg07x2}, we obtain the following assertions:

1) $H_1^N\setminus\{r_1,\ldots,r_{k-1}\} =
I_1^{N+1}\setminus\{r_1,\ldots,r_k\}$;

2) the double sum $\sum_{r_k\in I_{N+1}}\sum_{\{(l_h,r_h)\}_{h=1}^{k-1}\in [H_1,\ldots,H_N]_{k-1}}$ in \eqref{QuanAlg07k} represents the summation of the following:

$\bullet$ all  combinations $1\le l_1<\ldots<l_{k-1}\le N$,

$\bullet$ all $r_k\in I_{N+1}$, 

$\bullet$ all $r_{k-1}$ in $H_{l_{k-1}}^N$ (equivalently, all $r_{k-1}$  in $I_{l_{k-1}}^{N+1}$ excluding $r_k$), 

$\bullet$ all $r_{k-2}$ in $H_{l_{k-2}}^N\setminus \{r_{k-1}\}$ (equivalently, all $r_{k-2}$ in $I_{l_{k-2}}^{N+1}$ excluding $r_h$ with $h\in\{k-1,k\}$), 

$\ldots,\ldots$, 

$\bullet$ all $r_1$ in $H_{l_1}^N\setminus \{r_2,\ldots,r_{k-1}\}$ (equivalently, all $r_1$ in  $I_{l_1}^{N+1}$ excluding $r_h$ with $h\in\{2,\ldots,k\}$.\\
In other words,
\[\sum_{r_k\in I_{N+1}}\sum_{\{(l_h,r_h)\}_{h=1}^{k-1}
\in[H_1,\ldots,H_N]_{k-1}}=\sum_{\{(l_h,r_h)\}_{h=1}^k
\in [I_1,\ldots,I_{N+1}]_k,\ l_k=N+1}
\]

3) for any $s\in(0,N]$,
\begin{align*}
&c_{H_1,\ldots,H_N}(\{(l_h,r_h)\}_{h=1}^{k-1};s)\notag\\ \overset{\eqref{QuanAlg06b3}}
=&\begin{cases}  
\big\vert H_s^N(r_p)\setminus\{r_{p+1},\ldots,r_{k-1}\} \big \vert,&\text{ if }s=l_p\text{ with }p\in(0,k)\\ 
\big\vert H_s^N\setminus\{r_{p+1},\ldots,r_{k-1}\}\big \vert,&\text{ if }l_p<s<l_{p+1}\text{ with }p\in[0,k)
\end{cases}\\
=&\begin{cases}  
\big\vert I_s^N(r_p)\setminus\{r_{p+1},\ldots, r_{k-1},r_k\} \big \vert,&\text{ if }s=l_p\text{ with }p\in(0,k)\\ 
\big\vert I_s^N \setminus\{r_{p+1},\ldots, r_{k-1}, r_k\}\big \vert,&\text{ if }l_p<s<l_{p+1}\text{ with }p\in[0,k)
\end{cases}\\
\overset{\eqref{QuanAlg06b3}}=&c_{ I_1,\ldots, I_{N+1}}(\{(l_h,r_h)\}_{h=1}^k;s) 
\end{align*}

4) for any $s=N+1=l_k$,
\[c_{ I_1,\ldots, I_{N+1}}(\{(l_h,r_h)\}_{h=1}^k;s) \overset{\eqref{QuanAlg06b3}}=I_{N+1}^{N+1}(r_k)=I_{N+1}(r_k)\]

Summing up, one gets \eqref{QuanAlg07n} and thereby  the proof is completed.             \end{proof}

\bigskip\bigskip

\section{The case of ${\bf p}_k$'s being identity}\label{qm-s4}\bigskip

When all ${\bf p}_k$ are equal to the identity operator, both \eqref{QuanAlg01a} and \eqref{QuanAlg01b} become trivial, while the first $(q,m)$--commutation relation, i.e., the equality in \eqref{QuanAlg01c}, reduces to the standard $q$--commutation relation.
Therefore, the algebra defined in Definition \ref{QuanAlg} corresponds to the standard quon algebra with the {\bf scalar} parameter $q$ falling within the interval $[-1, 1]$.
 
For any $n\in{\mathbb N}^*$ and $\varepsilon\in \{-1,1\}^n$, we denote
\begin{align}\label{QuanAlg10a}
m_\varepsilon&:=\min\big\{\big\vert\varepsilon^{-1} (\{1\})\big\vert,\big\vert\varepsilon^{-1}(\{-1\})\big \vert\big\} \notag\\ 
F_\varepsilon(k)&:=\begin{cases}
\big\{\{(l_h,r_h)\}_{h=1}^k:\, l_1<\ldots<l_k,\,\ \vert\{r_1,\ldots,r_k\}\vert=k \\
\hspace{1.5cm} \varepsilon(l_h)=-1,\, \varepsilon(r_h)=1,\, r_h>l_h,\,\forall h\le k\big\}, &\text{ if } k\in (0,m_\varepsilon]\\
\emptyset\,,&\text{ if } k>m_\varepsilon
\end{cases}
\end{align}
In commonly used terminology, as often seen in references such as \cite{ManSchSev2007}:

$\bullet$ every element in $F_\varepsilon(k)$ is referred to as a {\it Feynman diagram} (related to $\varepsilon$) of the degree $k$;

$\bullet$ every $l_h$ is called a {\it paired ${\mathcal A}$--index}, and every $r_h$ is termed a {\it paired ${\mathcal C}$--index}; 

$\bullet$ every element in the set $\varepsilon^{-1} (\{-1\}) \setminus\{l_h:h\in (0,k]\}$ is named as a {\it unpaired ${\mathcal A}$--index}; 

$\bullet$ every element in the set $\varepsilon^{-1}(\{1\} \setminus\{r_h:h\in (0,k]\}$ is referred to as a {\it unpaired ${\mathcal C}$--index}.

For any Feynman diagram $\gamma:= \{(l_h,r_h)\}_{h=1}^k\in\cup_{n\ge1} F_\varepsilon(n)=\cup_{n=1}^{m_\varepsilon} F_\varepsilon(n)$, the following three specific characterizations of $\gamma$ are crucial for stating and proving the $q$--Wick's theorem, as well as our main result in this section:

$\bullet$ the {\it restricted crossing number} of $\gamma$,
\begin{align}\label{QuanAlg10b1a}
d_1(\gamma):=\begin{cases}\sum_{h=1}^k \big\vert\{p: l_p<l_h <r_p<r_h\} \big\vert,&\text{ if } k\le m_\varepsilon\text{ and }\gamma=\{(l_h,r_h)\}_{h=1}^k\\
0, &\text{ if }k> m_\varepsilon\end{cases}
\end{align}
i.e. the sum, taken over all $h\in(0,k]$, of the cardinality of the set containing all $p\in(0,h]$ such that $r_p$  falls within the interval $(l_h,r_h)$;

$\bullet$ the {\it degenerate crossing number} of $\gamma$,
\begin{align}\label{QuanAlg10b1b}
d_2(\gamma):=\begin{cases}\sum_{h=1}^k \big\vert \big\{(l_h,r_h)\cap \gamma^c\big\}\big\vert,&\text{ if } k\le m_\varepsilon\text{ and }\gamma=\{(l_h,r_h)\}_{h=1}^k\\
0, &\text{ if }k> m_\varepsilon\end{cases}
\end{align}
here, $\gamma^c$ is defined as the {\bf set of all unpaired indices}, i.e., $\gamma^c:=(0,n]\setminus \{l_h,r_h:h\in(0,k]\}$; in other words, the degenerate crossing number of $\gamma$ is the sum, taken over all $h\in(0,k]$, of unpaired indices that locate between $l_h$ and $r_h$;

$\bullet$ the {\it length} of $\gamma$,
\begin{align}\label{QuanAlg10b1c}
d_3(\gamma):=&\sum_{s\in \varepsilon^{-1}(\{-1\}) \setminus\{l_h:h\in (0,k]\}} \big\vert (s,n]\cap \big(\varepsilon^{-1}(\{1\}) \setminus\{r_h:h\in (0,k]\}\big)\big\vert
\end{align}
i.e. the sum, taken over all $l$ in the set of unpaired ${\mathcal A}$--indices, of the cardinality of the set containing all unpaired ${\mathcal C}$--indices situated to the right of $l$.

Wick's theorem for the standard quon algebra, as described in \cite{ManSchSev2007} and related references, can be stated as follows: 

{\it Consider a (pre--)Hilbert space ${\mathcal H}$, and let $G_q({\mathcal H})$ be the *--algebra generated by the creation and annihilation operators defined on the $q$--Fock space over ${\mathcal H}$, denoted as $A^+(f)$ and $A(g)$, where $f,g\in {\mathcal H}$. Then, the following statement holds for any $n\in{\mathbb N}^*$ and $\varepsilon\in\{-1,1\}^n$, and for any $\{g_k:\,k\in{\mathbb N}\}\subset {\mathcal H}$:
\begin{align}\label{QuanAlg10}
&A^{\varepsilon(1)}(g_1)\ldots A^{\varepsilon(n)}(g_n)
\notag\\
=&q^{\sum_{s\in \varepsilon^{-1}(\{-1\}) } \big\vert (s,n]\cap \varepsilon^{-1}(\{1\}) \big\vert }A^+_{\bf g}(\varepsilon^{-1}(\{1\}))A_{\bf g} (\varepsilon^{-1}(\{-1\}))\notag\\
+&\sum_{k\ge1}\sum_{\gamma:=\{(l_h,r_h)\}_{h=1}^k\in F_\varepsilon(k)}q^{\pi_\varepsilon(\gamma)} \prod_{h=1}^k\langle g_{l_h},g_{r_h}\rangle \notag\\
&A^+_{\bf g}(\varepsilon^{-1}(\{1\}) \setminus \{r_1,\ldots,r_k\})A_{\bf g} (\varepsilon^{-1}(\{-1\}) \setminus\{l_1,\ldots,l_k\})
\end{align} 
Hereinafter, for any $k\in{\mathbb N}^*$ and any Feynman diagram $\gamma:=\{(l_h,r_h)\}_{h=1}^k$, we denote  $\pi_\varepsilon(\gamma):= \pi(\{(l_h,r_h)\}_{h=1}^k) := d_1(\gamma)+ d_2(\gamma)+d_3(\gamma)$.}

\begin{remark}\label{remQuanAlg15} 1) The increasing property of $l_h$'s tells us
\begin{align*}
\sum_{h=1}^k\big\vert\{p:l_p<l_h <r_p<r_h\} \big\vert
&=\sum_{h=1}^k\big\vert\{p<h:l_p<l_h <r_p<r_h\} 
\big\vert\\
&=\sum_{h=1}^k\big\vert\{p<h:r_p\in(l_h,r_h)\}\big\vert
\end{align*}

2) For any Feynman diagram $\gamma=\{(l_h,r_h)\} _{h=1}^k$, it is common to use the notation $\pi(\{(l_h,r_h)\}_{h=1}^k):= \pi(\gamma)$.

3) To obtain the normally ordered form of a word like $A^{\varepsilon(1)}(g_1)\ldots A^{\varepsilon(n)} (g_n)$, we can adopt the default convention of setting $\varepsilon(1)=-1$ and $\varepsilon(n)=1$.        

4) For any $\varepsilon\in\{-1,1\}^n$ such that $\varepsilon(1)=-1$ and $\varepsilon(n)=1$, 
there must exist uniquely $N\in(0,n]$ and $\{m_k\}_{k=1}^{N}$ satisfying the following:

$\bullet$ $1=:m_1<m_2<\ldots<m_N< n$;

$\bullet$ by denoting $I_j:=(m_j, m_{j+1}]$ for any $j\in(0,N]$ (here and in what follows, for clarity, $m_{N+1}:=n$), and by using the notations introduced in \eqref{QuanAlg04x},
\begin{align}\label{QuanAlg10c1}
&\varepsilon^{-1}(\{-1\})=\{m_j:j\in(0,N]\},\quad
\varepsilon^{-1}(\{1\})=I_1^N\notag\\
&(m_j,n]=(m_j,m_{N+1}],\quad (m_j,n]\cap \varepsilon^{-1}(\{1\})=I_j^N,\qquad \forall j\in (0,N]
\end{align}

$\bullet$ for any $r\in\varepsilon^{-1}(\{-1\})$ and $l=m_j\in\varepsilon^{-1}(\{-1\})$, the inequality $l<r$
holds if and only if $r\in I_j^N$;

$\bullet$ the sum $\sum_{\gamma:=\{(l_h,r_h)\} _{h=1}^k\in F_\varepsilon(k)}$ can be reformulated as follows:
\begin{align}\label{QuanAlg10c2}
\sum_{\gamma:=\{(l_h,r_h)\}_{h=1}^k\in F_\varepsilon(k)}
&=\sum_{1\le j_1<\ldots<j_k\le N\atop r_h>m_{j_h},\,\varepsilon(r_h)=1,\ \forall h,\, \vert\{r_1,\ldots,r_k\}\vert=k}\notag\\
&=\sum_{1\le j_1<\ldots<j_k\le N, r_k\in I_{j_k}^N \atop r_{k-1}\in I_{j_{k-1}}^N\setminus\{r_k\},\ldots, \,r_1\in I_{j_1}^N\setminus\{r_2,\ldots, r_k\}}
\end{align}           \end{remark}

With the help of these results, we can rewrite the formula \eqref{QuanAlg10} as follows:
\begin{align}\label{QuanAlg10c}
&A^{\varepsilon(1)}(g_1)\ldots A^{\varepsilon(n)}(g_n)\notag\\
=&q^{\sum_{s\in \varepsilon^{-1}(\{-1\}) } \big\vert (s,n]\cap \varepsilon^{-1}(\{1\}) \big\vert }\cdot A^+_{\bf g} (\varepsilon^{-1}(\{1\})\cdot A_{\bf g}(\varepsilon^{-1}(\{-1\})\notag\\
+&\sum_{k\in (0,N]} \sum_{1\le j_1<\ldots<j_k\le N, r_k \in I_{j_k}^N\atop r_{k-1}\in I_{j_{k-1}}^N\setminus \{r_k\},\ldots,\,r_1\in I_{j_1}^N\setminus\{r_2,\ldots, r_k\} }  \prod_{h\in(0,k]} \langle g_{m_{j_h}}, g_{r_h}\rangle \notag\\ 
&\cdot q^{\pi(\{(m_{j_h},r_h)\}_{h=1}^k)}\cdot 
A^+_{\bf g}(\varepsilon^{-1}(\{1\}) \setminus \{r_1,\ldots,r_k\})\notag\\ 
&\cdot A_{\bf g} (\varepsilon^{-1}(\{-1\}) \setminus\{m_{j_1},\ldots,m_{j_k}\})
\end{align}

On the other hand, the formula \eqref{QuanAlg06ax} yields, by taking all ${\bf p}_k$ as the identity,
\begin{align}\label{QuanAlg10d}
&A(g_{m_1})A^+_{\bf g}(I_1)A(g_{m_2})A^+_{\bf g} (I_2)\ldots A(g_{m_N})A^+_{\bf g}(I_N)\notag\\
=&q^{\sum_{j\in(0,N]}\vert I_j^N\vert}\cdot A^+_{\bf g} (I_1^N)\cdot A_{\bf g}(\{m_j:j\in(0,N]\}) \notag\\
+&\sum_{k\in(0,N]}\sum_{1\le j_1<\ldots<j_k\le N, r_k \in I_{j_k}^N\atop r_{k-1}\in I_{j_{k-1}}^N\setminus \{r_k\},\ldots,\,r_1\in I_{j_1}^N\setminus\{r_2,\ldots, r_k\} }  \prod_{h\in(0,k]} \langle g_{m_{j_h}}, g_{r_h}\rangle\notag\\ 
&\cdot q^{c(\{(m_{j_h}+1,r_h)\}_{h=1}^k)}\cdot A^+_{\bf g}(\varepsilon^{-1}(\{1\}) \setminus \{r_1,\ldots,r_k\})\notag\\ 
&\cdot A_{\bf g} (\varepsilon^{-1}(\{-1\}) \setminus\{m_{j_1}, \ldots,m_{j_k}\})
\end{align}
with
\begin{align}\label{QuanAlg10e}
c(\{(m_{j_h},r_h)\}_{h=1}^k):=\sum_{s\in(0,N]}c_{I_1,\ldots, I_N}(\{(l_h,r_h)\}_{h=1}^k;s) \Big\vert_{l_h= m_{j_h},\,\forall h}
\end{align}

Now we are prepared to state and prove the  main result of this section.

\begin{theorem}\label{QuanAlg15} For any $n\in{\mathbb N}^*$ and $\varepsilon\in\{-1,1\}^n$, using the notations introduced earlier for $\{m_j,I_j:j\in(0,N]\}$ in point 4) of Remark \ref{remQuanAlg15}, we can make the following assertion:
For any $k\in(0,N]$ and $1\le j_1<\ldots<j_k\le N$, along with $r_h \in I_{j_h}^N$, where $h\in(0,k]$ such that $\vert \{r_1,\ldots,r_k\}\vert=k$, 
\begin{align}\label{QuanAlg15a0}
\sum_{s\in \varepsilon^{-1}(\{-1\}) } \big\vert (s,n]\cap \varepsilon^{-1}(\{1\}) \big\vert=\sum_{j\in(0,N]}\vert I_j^N\vert
\end{align}
and
\begin{align}\label{QuanAlg15a}
\pi(\{(m_{j_h},r_h)\}_{h=1}^k)=c(\{(m_{j_h},r_h)\}_{h=1}^k)
\end{align}
Consequently, for any $q\in[-1,1]$ and $\{g_k:\,k\in{\mathbb N}\}\subset {\mathcal H}$,
\begin{align}\label{QuanAlg15b}
\text{the expression in \eqref{QuanAlg10c}}\,=\,
\text{the expression in \eqref{QuanAlg10d}}
\end{align}
\end{theorem}
\begin{proof} By comparing the expressions in \eqref{QuanAlg10c} and  \eqref{QuanAlg10d}, we can easily deduce the equality in \eqref{QuanAlg15b} from \eqref{QuanAlg15a0} and \eqref{QuanAlg15a}. Moreover, \eqref{QuanAlg15a0} can be easily obtained by observing that the expression on the left--hand side of \eqref{QuanAlg15a0} is equals to  
\[\sum_{j=1}^N \big\vert (m_j,n]\cap \varepsilon^{-1}(\{1\}) \big\vert\] 
where, $(m_j,n]\cap \varepsilon^{-1}(\{1\})$ is nothing else but $I_j^N$ by definition.
Consequently, the main focus of proving Theorem \ref{QuanAlg15} centres on establishing \eqref{QuanAlg15a} and we'll complete this in several steps.

Our first step is to calculate  $c(\{(m_{j_h},r_h)\} _{h=1}^k)$. By adopting the convention $m_{j_{k+1}}:=n+1$ and noting that
\begin{align*}
\sum_{s\in(0,N]\setminus\{j_1,\ldots,j_p\}}=
\sum_{s<j_1}+\sum_{p=1}^{k-1}\sum_{j_p<s<j_{p+1}}+\sum_{s>j_k}
\end{align*}
and
\[I_s^n \setminus\{r_{p+1},\ldots, r_k\}=I_s^n \text { for any }s>j_k\]
we find that

\begin{align}\label{QuanAlg10f}
&c(\{(m_{j_h},r_h)\}_{h=1}^k)=c(\{(l_h,r_h)\}_{h=1}^k) \Big\vert_{l_h=m_{j_h},\,\forall h}\notag\\
\overset{\eqref{QuanAlg10e}}{:=}& \sum_{s\in(0,N]} c_{I_1,\ldots, I_N}(\{(l_h,r_h)\} _{h=1}^k;s)\Big\vert_{l_h=m_{j_h},\,\forall h}\notag\\
\overset{\eqref{QuanAlg06b3}}{:=}&
\sum_{s\in(0,N]}\begin{cases}  
\big\vert I_s^n(r_p)\setminus\{r_{p+1},\ldots,r_k\}\big \vert,&\text{if }s=m_{j_p}\text{ with }p\in(0,k]\\ 
\big\vert I_s^n \setminus\{r_{p+1},\ldots, r_k\} \big\vert,&\text{if }m_{j_p}<s<m_{j_{p+1}} \text{ with }p\in [0,k]\end{cases}\notag\\
=&\sum_{p=1}^{k}\big\vert I_{j_p}(r_p)^n \setminus \{r_{p+1},\ldots, r_k\} \big\vert+ \sum_{s\in(0,N] \setminus\{j_1,\ldots,j_p\}}\big\vert I_s^n \setminus\{r_{p+1},\ldots, r_k\} \big\vert\notag\\
=&\sum_{h=1}^k\vert I_{j_h} ^N(r_h)\setminus \{r_{h+1}, \ldots,r_k\} \vert+\sum_{s<j_1}\big\vert I_s^N\setminus\{r_1,\ldots, r_k\}\big\vert \notag\\
+&\sum_{p=1}^{k-1}\sum_{s\in(j_p,j_{p+1})} \big\vert I_s^N \setminus\{r_{p+1},\ldots, r_k\} \big\vert+\sum_{s>j_k}\big\vert I_s^N \big\vert 
\end{align}

As the second step, we'll do some preparation for computing $\pi(\{(m_{j_h},r_h)\}_{h=1}^k)$, which is defined as $d_1(\gamma)+d_2(\gamma)+d_3(\gamma)$, where $\gamma= \{(l_h,r_h)\}_{h=1}^k$ with $l_h=m_{j_h}$ for all $h\in(0,k]$.

Let's introduce the following notations:
\begin{align}\label{QuanAlg10b1y}
L_0&:=\emptyset=:R_0\,;\ \, R_h:=\{r_p:p\in(0,h]\} \text{ and } L_h:=\{m_{j_p}:p\in(0,h]\},\ \ \forall h\in(0,k] \notag\\
E^c&:=[1,n]\setminus E=[m_1,n]\setminus E,\qquad \forall E\subset [1,n]
\end{align} 
It is easy to see that
\begin{align}\label{QuanAlg10b1x}
&R_k\cap L_k=\emptyset;\quad \{m_{j_h},r_h:h\in(0,k]\}=
R_k\cup L_k\notag\\
&L_h\subset L_k\subset \varepsilon^{-1}(\{-1\}) \text{ and }R_h\subset R_k\subset\varepsilon^{-1}(\{1\}),
\quad \forall h\in(0,k]\notag\\
&\varepsilon^{-1}(\{1\})\cap L_k^c =\varepsilon^{-1} (\{1\}); \quad \varepsilon^{-1}(\{-1\})\cap R_k^c= \varepsilon^{-1}(\{-1\})\notag\\
&\gamma^c:=[m_1,n]\setminus\{m_{j_h},r_h:h\in(0,k]\}=R_k^c\cap L_k^c\notag\\
&\ \quad=\big(\varepsilon^{-1}(\{1\}) \cap R_k^c\big) \cup \big(\varepsilon^{-1}(\{-1\})\cap L_k^c\big)
\end{align} 
Moreover, $d_i(\gamma)$'s given in 
\eqref{QuanAlg10b1a}, \eqref{QuanAlg10b1b} and \eqref{QuanAlg10b1c} can be rewritten to
\begin{align}\label{QuanAlg10b1z}
d_1(\gamma)&:=\sum_{h=1}^k \big\vert\{p: m_{j_p}<m_{j_h} <r_p<r_h\} \big\vert\notag\\
d_2(\gamma)&:=\sum_{h=1}^k\big\vert(m_{j_h},r_h)\cap \gamma^c\big\vert=\sum_{h=1}^k \big\vert(m_{j_h},r_h) \cap R_k^c\cap L_k^c\big\vert\notag\\
d_3(\gamma)&:=\sum_{s\in\varepsilon^{-1}(\{-1\}) \setminus L_k}\big\vert(s,n]\cap \big(\varepsilon^{-1} (\{1\})\setminus R_k\big) \big\vert
\end{align}
Because of the strictly increasing property of $m_{j_h}$'s (equivalently, $j_h$'s), and by applying the definition of $I_h$'s, $I_h^N$'s and $I_h^N(r)$'s,
we find that for any $h\in(0,k]$,
\begin{align}\label{QuanAlg10b1w}
&R_h\cap I_{j_h}^N(r_h)=R_{h-1}\cap I_{j_h}^N(r_h)\notag\\ 
&\{p: m_{j_p}<m_{j_h} <r_p<r_h\}=\{p<h: r_p\in (m_{j_h},r_h)\}=\{p<h: r_p\in I_{j_h}^N(r_h)\}\notag\\
&(m_{j_h},r_h)\setminus R_k=(m_{j_h},r_h)\cap R_k^c=\big( (m_{j_h},r_h)\cap R_k^c \cap L_k^c\big) \cup \big( (m_{j_h},r_h)\cap R_k^c \cap L_k\big)\notag\\
&(m_{j_h},r_h)\cap\gamma^c\overset{\eqref{QuanAlg10b1x}} =\big((m_{j_h},r_h)\cap \varepsilon^{-1}(\{1\})\cap R_k^c\big)\cup \big((m_{j_h},r_h)\cap \varepsilon^{-1}(\{-1\})\cap L_k^c\big)\notag\\
&\hspace{2.4cm}=\big(I_{j_h}^N(r_h)\setminus R_k\big) \cup\big((m_{j_h},r_h)\cap \varepsilon^{-1}(\{-1\})\cap L_k^c\big)
\end{align}
and consequently,
\begin{align}\label{QuanAlg10b1t}
&\vert\{p: m_{j_p}<m_{j_h} <r_p<r_h\}\vert= \vert\{p<h: r_p\in I_{j_h}^N(r_h)\}\vert
=\vert I_{j_h}^N(r_h)\cap R_h\vert\notag\\
&\vert(m_{j_h},r_h)\cap \gamma^c \vert= \vert I_{j_h}^N(r_h)\setminus  R_k\vert+ \vert(m_{j_h},r_h) \cap \varepsilon^{-1}(\{-1\}) \cap L_k^c\vert
\end{align}

Our third step is to compute the sum $d_1(\gamma)+d_2(\gamma)$. To do this, we make use of the following set manipulation formula: for any set $X$ and its subsets $E,A,B$,
\begin{align}\label{QuanAlg12d}
E\setminus(A\setminus B)&=E\setminus(A\cap B^c)= E\cap (A\cap B^c)^c\notag\\
&= E\cap (A^c\cup B)=\big(E\setminus A\big)\cup\big(E\cap B\big)
\end{align}
here, $D^c:=X\setminus D$ for any $D\subset X$.

For any $h\in(0,k]$, choosing $E:=I_{j_h}^N(r_h)$, $A:=R_k$, and $B:=R_h$ (thus, $\{r_{h+1},\ldots,r_k\}= A\setminus B$), the equality \eqref{QuanAlg12d} confirms that 
\begin{align}\label{QuanAlg12a}
I_{j_h}^N(r_h)\setminus\{r_{h+1},\ldots,r_k\}=
\big(I_{j_h}^N(r_h)\setminus  R_k\big)\cup\big( I_{j_h}^N(r_h)\cap R_h\big)
\end{align}
and consequently
\begin{align}\label{QuanAlg12b}
\vert I_{j_h}^N(r_h)\setminus\{r_{h+1},\ldots,r_k\} \vert= \vert I_{j_h}^N(r_h)\setminus  R_k\vert +\vert I_{j_h}^N(r_h)\cap R_h\vert
\end{align} 
Therefore, 
\begin{align}\label{QuanAlg12c}
&d_1(\gamma)+d_2(\gamma)\notag\\
\overset{\eqref{QuanAlg10b1z}}{:=}&
\sum_{h=1}^k\big\vert\{p: m_{j_p}<m_{j_h} <r_p<r_h\} \big\vert   +\sum_{h=1}^k \big\vert(m_{j_h},r_h)\cap\gamma^c\big\vert \notag\\
\overset{\eqref{QuanAlg10b1t}}=&\sum_{h=1}^k
\Big(\big\vert I_{j_h}^N(r_h)\cap R_h \big\vert+ \big\vert I_{j_h}^N(r_h)\setminus  R_k\big\vert+ \big\vert(m_{j_h},r_h) \cap \varepsilon^{-1}(\{-1\}) \cap L_k^c\big\vert \Big)\notag\\
\overset{\eqref{QuanAlg12b}}=&\sum_{h=1}^k \big\vert I_{j_h}^N(r_h)\setminus\{r_{h+1}, \ldots,r_k\}\big\vert +\sum_{h=1}^k\big\vert(m_{j_h},r_h)\cap \varepsilon^{-1}(\{-1\}) \cap L_k^c\big\vert
\end{align}	

As the fourth step, we'll demonstrate that
\begin{align}\label{QuanAlg14a}
&d_3(\gamma)+\sum_{h=1}^k\vert(m_{j_h},r_h) \cap \varepsilon^{-1}(\{-1\}) \cap L_k^c\vert\notag\\
=&\sum_{s<j_1}\big\vert I_s^N\setminus\{r_1,\ldots, r_k\}\big\vert+\sum_{p=1}^{k-1}\sum_{s\in(j_p,j_{p+1})} \big\vert I_s^N \setminus\{r_{p+1},\ldots, r_k\} \big\vert+\sum_{s>j_k}\big\vert I_s^N \big\vert 
\end{align}  

It is easy to know that the formula \eqref{QuanAlg12d} straightforwardly yields the following analogies of \eqref{QuanAlg12a} and \eqref{QuanAlg12b}: for any $p\in(0,k]$ and $s\in (0,N]$
\begin{align}\label{QuanAlg14b}
&I_s^N\setminus\{r_{p+1},\ldots,r_k\}= I_s^N\setminus(R_k\setminus R_p)\notag\\
=&\big(I_s^N\setminus R_k\big)\cup\big( I_s^N\cap R_h\big)=\big(I_s^N\cap R^c_k\big)\cup\big( I_s^N\cap R_h\big)
\end{align}
and, thanks to the fact $R_h\subset R_k$ or any $h\le k$,
\begin{align}\label{QuanAlg14c}
\vert I_s^N\setminus\{r_{h+1},\ldots,r_k\}\vert= 
\big\vert I_s^N\setminus R_k\big\vert+\big\vert I_s^N\cap R_h\big\vert
\end{align}
Therefore, the expression on the right--hand side of  \eqref{QuanAlg14a} equals to
\begin{align}\label{QuanAlg14d0}
\sum_{s<j_1}\big\vert I_s^N\setminus R_k\big\vert &+\sum_{p=1}^{k-1}\sum_{s\in(j_p,j_{p+1})} \big(\big\vert I_s^N \setminus R_k \big\vert+\big\vert I_s^N \cap R_p \big\vert\big)\notag\\
&+\sum_{s>j_k} \big(\big\vert I_s^N\setminus R_k \big\vert+\big\vert I_s^N\cap R_k \big\vert \big)
\end{align}
Since
\begin{align*}
&\sum_{s<j_1}\big\vert I_s^N\setminus R_k\big\vert +\sum_{p=1}^{k-1}\sum_{s\in(j_p,j_{p+1})} \big\vert I_s^N \setminus R_k \big\vert+\sum_{s>j_k} \big\vert I_s^N\setminus R_k \big\vert \\
=&\sum_{s\in(0,N]\setminus\{j_1,\ldots,j_k\}}\big\vert I_s^N\setminus R_k\big\vert
\end{align*}
the expression \eqref{QuanAlg14d0} is equal to
\begin{align}\label{QuanAlg14d}
\sum_{s\in(0,N]\setminus\{j_1,\ldots,j_k\}}\big\vert I_s^N\setminus R_k\big\vert+\sum_{p=1}^{k-1} \sum_{s\in(j_p,j_{p+1})} \big\vert I_s^N \cap R_p \big\vert+\sum_{s>j_k} \big\vert I_s^N\cap R_k\big\vert
\end{align}

On the other hand, considering any $m_p\in \varepsilon^{-1}(\{-1\}) \setminus L_k$ (equivalently, $p\in (0,N]\setminus\{j_1,\ldots,j_k\}$ and so $(m_p,n]\cap \varepsilon^{-1}(\{1\})=I_p^N$),  we can express $d_3(\gamma)$ as:
\begin{align}\label{QuanAlg14e}
d_3(\gamma)\overset{\eqref{QuanAlg10b1z}}{:=}&\sum_{s\in\varepsilon^{-1}(\{-1\}) \setminus L_k}\big\vert(s,n]\cap \big(\varepsilon^{-1} (\{1\})\setminus R_k\big)\big\vert\notag\\ =&\sum_{s\in(0,N]\setminus\{j_1,\ldots,j_k\} } \big\vert I_s^N\setminus R_k\big\vert
\end{align}
Therefore, the expression \eqref{QuanAlg14d} (i.e., the expression on the right--hand side of \eqref{QuanAlg14a}) can be written as
\[
d_3(\gamma)+\sum_{p=1}^{k-1} \sum_{s\in(j_p,j_{p+1})} \big\vert I_s^N \cap R_p \big\vert+\sum_{s>j_k} \big\vert I_s^N\cap R_k\big\vert
\]
Hence, the equality \eqref{QuanAlg14a} is equivalent to the following equality:
\begin{align}\label{QuanAlg14f}
\sum_{h=1}^k\big\vert(m_{j_h},r_h)\cap
\varepsilon^{-1} (\{-1\})\cap L_k^c\big\vert=\sum_{p=1}^{k-1} \sum_{s\in (j_p,j_{p+1})}\big\vert I_s^N\cap R_p\big\vert +\sum_{s>j_k} \big\vert I_s^N\cap R_k\big\vert
\end{align}

To prove the \eqref{QuanAlg14f}, let's introduce
\begin{align}\label{QuanAlg14g}
x_k&:=\text{the expression on the left--hand side of \eqref{QuanAlg14f}}\notag\\
y_k&:=\text{the expression on the right--hand side of \eqref{QuanAlg14f}}\notag\\
x_{k,h}&:=\vert(m_{j_h},r_h)\cap\varepsilon^{-1}(\{-1\}) \cap L_k^c\vert,\qquad \forall h\in(0,k] \notag\\
y_{k,h}&:=\sum_{s\in (j_h,j_{h+1})}\big\vert I_s^N\cap R_h\big\vert,\qquad \forall h\in(0,k) \notag\\
y_{k,k}&:= \sum_{s>j_k} \big\vert I_s^N\cap R_k\big\vert 
\end{align}
then, it is clear that
\[x_{1,1}=x_1,\quad y_{1,1}=y_1;\qquad x_k=\sum_{h=1}^kx_{k,h},\quad y_k=\sum_{h=1}^ky_{k,h} \]

In the case where $k=1$, since $\big\vert I_l^N\cap \{r_1\}\big\vert=\chi_{(m_l} (r_1)$, where $\chi_{(q} (p):=\begin{cases}
1,&\text{ if }p>1\\ 0,&\text{ if }p\le q
\end{cases}$ for any $p,q\in{\mathbb R}$, we can apply the definitions of $x_1$ and $y_1$ to obtain the following equality:
\begin{align*}
x_1=\vert(m_{j_1},r_1)\cap(\varepsilon^{-1}(\{-1\})
\setminus\{s_1\})\vert=\vert\{s>j_1:m_s<r_1\}\vert=
\sum_{s>j_1} \big\vert I_s^N\cap \{r_1\}\big\vert=y_1
\end{align*}	
Now, let's move on to the case of $k\ge 2$. By taking into account the fact
\begin{align}\label{QuanAlg14j1}
y_{k,h}&:=\sum_{s\in (j_h,j_{h+1})}\big\vert I_s^N\cap \{r_1,\ldots,r_h\}\big\vert\notag\\
&=\sum_{p=1}^h\big\vert\{s\in (j_h,j_{h+1}):m_s<r_p \} \big\vert,\quad \forall h\in(0,k)
\end{align}
and by adopting the convention $l_{k+1}:=N+1$, we have
\begin{align}\label{QuanAlg14j2}
y_{k,k}&:=\sum_{s>j_k}\big\vert I_s^N\cap \{r_1,\ldots,r_k\}\big\vert=\sum_{p=1}^k\big\vert
\{s\in(j_k,j_{k+1}):m_s<r_p\} \big\vert
\end{align}
On the other hand,
\begin{align}\label{QuanAlg14j}
x_{k,h}&:=\vert(m_{j_h},r_h)\cap(\varepsilon^{-1}
(\{-1\}) \setminus\{r_1,\ldots,r_k\} \vert\notag\\
&=\vert\{s>j_h:s\notin\{j_p:p\in(h,k]\}:m_s<r_h\}\vert, \qquad \forall h\in(0,k]
\end{align}
Therefore, by summing over all $ h\in(0,k]$ and combining the formulae \eqref{QuanAlg14j1}, \eqref{QuanAlg14j2} and \eqref{QuanAlg14j}, we can conclude that
\begin{align*}
y_k&=\sum_{h=1}^ky_{k,h}=\sum_{h=1}^{k}\sum_{p=1}^h\big\vert\{s\in (j_h,j_{h+1}):m_s<r_p \} \big\vert\notag\\
& =\sum_{p=1}^{k}\sum_{h=p}^k\big\vert\{s\in (j_h,j_{h+1}):m_s<r_p \} \big\vert =\sum_{p=1}^{k} \big\vert\{s>j_p:m_s<r_p \} \big\vert\notag\\
&=\sum_{p=1}^{k}x_{k,p}=x_k
\end{align*}

Finally, as the last step, by combining \eqref{QuanAlg12c} and \eqref{QuanAlg14a}, we obtain that
\begin{align}\label{QuanAlg16}
&d_1(\gamma)+d_2(\gamma)+d_3(\gamma)\notag\\
=&\sum_{h=1}^k\vert I_{j_h} ^N(r_h)\setminus \{r_{h+1}, \ldots,r_k\} \vert+\sum_{s<j_1}\big\vert I_s^N\setminus\{r_1,\ldots, r_k\}\big\vert \notag\\
+&\sum_{p=1}^{k-1}\sum_{s\in(j_p,j_{p+1})} \big\vert I_s^N \setminus\{r_{p+1},\ldots, r_k\} \big\vert+\sum_{s>j_k}\big\vert I_s^N \big\vert \notag\\
=&\sum_{p=1}^{k}\big\vert I_{j_p}(r_p)^n \setminus\{r_{p+1},\ldots, r_k\} \big\vert+ \sum_{s\in(0,N]\setminus\{j_1,\ldots,j_p\}}\big\vert I_s^n \setminus\{r_{p+1},\ldots, r_k\} \big\vert
\end{align}
and which is nothing else but $c(\{(m_{j_h},r_h)\}_ {h=1}^k)$ in virtue of \eqref{QuanAlg10f}.
\end{proof}


\newpage
\begin{thebibliography}{9}                                        
\bibitem{accardi90}L. Accardi: An outline of quantum probability, (1990). www.researchgate.net/publication/280574426 \_An\_outline\_of\_quantum\_probability 
	
\bibitem{accardi95}L. Accardi: Quantum Probability theory: Stochastic Calculus and Conditional Expectations. {\it Trans. Moscow Math. Soc.}, {\bf 56}, pp.235--270 (1995).
	
\bibitem{AcKoVo98b}L. Accardi, S.V. Kozyrev, I. Volovich: Dynamical origins of $q$-deformations in QED and the stochastic limit, {\it Journal of Physics A, Math. Gen.}, {\bf 32}, pp.3485--3495 (1999).

\bibitem{aclu-qed}L. Accardi, Y.G. Lu: The Wigner semi-circle law in quantum electrodynamics. {\it Commun. Math. Phys.}, {\bf 180}, pp.605--632 (1996).

\bibitem{acluvo-kyoto}L. Accardi, Y.G. Lu, I. Volovich: The QED Hilbert Module and Interacting Fock Spaces. IIAS--Kyoto report (1997).
	
\bibitem{AsaiYosh2020}N. Asai, H. Yoshida: Deformed Gaussian Operators on Weighted $q$--Fock Spaces. {\it Journal of Stochastic Analysis}, v.{\bf 1}, n.{\bf 4} (2020). DOI: 10.31390/josa.1.4.06 
	
\bibitem{Blitvic2012}N. Blitvi\'c: The $(q,t)$--Gaussian process. {\it J. Funct. Anal.}, {\bf 263}, pp.3270--3305 (2012).
	
\bibitem{Bo-Kum-Spe97}M. Bo\.zejko, B. K\"ummerer, R. Speicher: $q$--Gaussian Processes: Non-commutative and Classical Aspects, {\it Commun. Math. Phys.}, {\bf 185}, pp.129--154 (1997).
	
\bibitem{Bo-Spe91}M. Bo\.zejko, R. Speicher: An example of a generalized Brownian motion. {\it Commun. Math. Phys.}, {\bf 137}, pp.519--531 (1991).
	
\bibitem{Bo-Spe96}M. Bo\.zejko, R. Speicher: Interpolations between bosonic and fermionic relations given by generalized Brownian motions. {\it Math. Z.}, {\bf 222}, pp.135--160 (1996).
	
\bibitem{BozeWyso2001}M. Bo\.zejko, J. Wysocza\'nski: Remarks on t-transformations of measures and convolutions, {\it Ann. Inst. H. Poincare Probab. Statist.}, {\bf 37}, pp.737--760 (2001).
	
\bibitem{BozeYosh2006}M. Bo\.zejko, H. Yoshida: Generalized $q$--deformed Gaussian random variables, {\it Banach Center Publ.}, {\bf 73}, pp.127--140 (2006).
	
\bibitem{Fiv92}D.I. Fivel: Interpolation between Fermi and Bose statistics using generalized commutators. {\it Phys. Rev. Lett.}, {\bf 65}, pp.3361--3364 (1990). Erratum {\bf 69}, p.2020 (1992).	
	
\bibitem {Fris-Bou70}U. Frisch, R. Bourret: Parastochastics, {\it J. Math. Phys.}, {\bf 11}, pp.364--390 (1970).
	
\bibitem {Gre91}O.W. Greenberg: Particles with small violations of Fermi or Bose statistics. {\it Phys. Rev. D}, {\bf 43}, pp.4111--4120 (1991). 
	
\bibitem {Gre2001}O.W. Greenberg: Theories of Violation of Statistics. {\it AIP Conference Proceedings}, v. {\bf 545}, Issue {\bf 1} (2001). DOI:10.1063/1.1337721 
	
\bibitem{GreHil99}O.W. Greenberg, R.C. Hilborn: Quon Statistics for Composite Systems and a Limit on the Violation of the Pauli Principle for Nucleons and Quarks, {\it Phys. Rev. Lett.}, v.{\bf 83}, n.{\bf 22}, pp.4460--4463 (1999).
	
\bibitem{JiKim2006}U.G. Ji, Y.Y. Kim: On a $q$--Fock and its unitary decomposition. {\it Bull. Korean Math. Soc.}, v.{\bf 43}, n.{\bf 1}, pp.53--62 (2006).
	
	
\bibitem{lu95}Y.G. Lu: On the interacting free Fock space and the deformed Wigner law, {\it Nagoya Math. J.}, v. {\bf 145}, pp.1--28 (1997).
	
\bibitem{lu-gifs}Y.G. Lu: Gaussian type interacting Fock space. {\it Inf. Dim. Anal. Quantum Prob. and Rel. Top.}, v.{\bf 11}, n.{\bf 4}, pp.475--494 (2008).
DOI:10.1142/S0219025708003300

\bibitem{lu2023c}Y.G. Lu: Some combinatorial aspects of $(q,m)-$Fock space, accepted for publication in J. of Combinatorics.


\bibitem{lu2023d}Y.G. Lu: Some Probabilistic aspects of $(q,2)-$Fock space, submitted to J. Stochastics.
	
\bibitem{ManSchSev2007}T. Mansour, M. Schork, S. Severrini:  Wick’s theorem for q-deformed boson operators, {\it Journal of Physics: A Mathematical and Theoretical}, v.{\bf 40}, n.{\bf 29}, pp.8393--8401 (2007). DOI:10.1088/1751-8113/40/29/014
	
\bibitem{Rand2019}H. Randriamaro: A Deformed Quon Algebra. {\it Communications in Mathematics}, v.{\bf 27}, Issue {\bf 2}, pp.103--112 (2019). DOI.org/10.2478/cm-2019-0010 
	
\bibitem{Zag92}D. Zagier: Realizability of a model in infinite statistics. {\it Commun. Math. Phys.}, {\bf 147}, pp.199--210 (1992).
	
\end{thebibliography}
\end{document}